\documentclass[a4paper,UKenglish, autoref, cleveref, thm-restate, nolineno]{socg-lipics-v2021}

\hideLIPIcs  

\graphicspath{{./figures/}}

\title{Minimum Height Drawings of Ordered Trees in Polynomial Time: Homotopy Height of Tree Duals}
\titlerunning{Minimum Height Drawings of Ordered Trees in Polynomial Time}

\author{Tim Ophelders}{Department of Information and Computing Science, Utrecht University, the Netherlands \and Department of Mathematics and Computer Science, TU Eindhoven, the Netherlands}{t.a.e.ophelders@uu.nl}{{https://orcid.org/0000-0002-9570-024X}}{This author was supported by the Dutch Research Council (NWO) under project no.\ VI.Veni.212.260.}

\author{Salman Parsa}{Scientific Computing and Imaging Institute, University of Utah, Salt Lake City, UT, USA}{sparsa@sci.utah.edu}{{https://orcid.org/0000-0002-8179-9322}}{This author was funded in part by the SLU Research Institute and by NSF grant CCF-1614562.}

\authorrunning{T. Ophelders and S. Parsa} 

\Copyright{Tim Ophelders and Salman Parsa} 



\nolinenumbers 

\EventEditors{Xavier Goaoc and Michael Kerber}
\EventNoEds{2}
\EventLongTitle{38th International Symposium on Computational Geometry (SoCG 2022)}
\EventShortTitle{SoCG 2022}
\EventAcronym{SoCG}
\EventYear{2022}
\EventDate{June 7--10, 2022}
\EventLocation{Berlin, Germany}
\EventLogo{socg-logo}
\SeriesVolume{224}
\ArticleNo{XX} 

\usepackage{hyperref}
\usepackage{mathtools}
\usepackage{xcolor}
\usepackage{comment}
\usepackage{tikz-cd}
\usetikzlibrary{decorations.pathmorphing}

\newcommand{\R}{\mathbb{R}}					

\newcommand{\from}{\colon}

\begin{document}

\maketitle
\begin{abstract}
We consider drawings of graphs in the plane in which vertices are assigned distinct points in the plane and edges are drawn as simple curves connecting the vertices and such that the edges intersect only at their common endpoints. There is an intuitive quality measure for drawings of a graph 
that measures the height of a drawing $\phi \from G\hookrightarrow\R^2$ as follows.
For a vertical line $\ell$ in $\mathbb{R}^2$, let the height of $\ell$ be the cardinality of the set $\ell \cap \phi(G)$. The height of a drawing of $G$ is the maximum height over all vertical lines. In this paper, instead of abstract graphs, we fix a drawing and consider plane graphs. In other words, we are looking for a homeomorphism of the plane that minimizes the height of the resulting drawing.
This problem is equivalent to the homotopy height problem in the plane, and the homotopic Fr\'echet distance problem. These problems were recently shown to lie in NP, but no polynomial-time algorithm or NP-hardness proof has been found since their formulation in 2009.
We present the first polynomial-time algorithm for drawing trees with optimal height. This corresponds to a polynomial-time algorithm for the homotopy height where the triangulation has only one vertex (that is, a set of loops incident to a single vertex), so that its dual is a tree.
\end{abstract}

\section{Introduction}
A tree $T$ is called an \emph{ordered tree} if for each vertex, a fixed cyclic ordering of its incident edges is given.
Let $T$ be an ordered tree and let $f\from |T|\to \mathbb{R}^2$ be a drawing of the tree, that is, a continuous injection from the underlying topological space of the tree to the plane, in which the clockwise order of edges around each vertex is as prescribed. Any ordered tree can be recovered from any of its drawings up to degree 2 nodes. Any two drawings of the same ordered tree can be obtained from one another using an orientation-preserving homeomorphism of the plane.
We are interested in drawings that minimize the height in the following sense. Given a drawing $\phi$ and a vertical line $\ell$, the \emph{height} of the line $\ell$ is defined as $H(\ell):=|\phi(T) \cap \ell|$. That is, the number of times that the line $\ell$ intersects the drawing, where vertical segments count as infinitely many intersections. The problem of drawing a tree $T$ with optimal height asks for a drawing $\phi\from |T|\to \mathbb{R}^2$ that minimizes the maximum height over all vertical lines. We call such a drawing an \textit{optimal height drawing}. We emphasize that our drawings are not necessarily straight-line. In fact, there exist instances for which any optimal drawing requires a bend in some edge. An example is given in Figure~\ref{fig:bent}. One can check that any optimal drawing of this tree requires a bend in some edge. Although we will consider only unweighted trees, the definition of height naturally extends to edge-weighted graphs. Already in the case of weighted trees with only one vertex of degree at least three, an optimal drawing might even require an edge to form a spiral. Figure~\ref{fig:spiral} depicts an instance whose optimal drawing requires a spiral according to a computer-assisted enumeration of its drawings.
We do not know whether unweighted trees also require spiraling edges.

\begin{figure}
    \centering
    \includegraphics{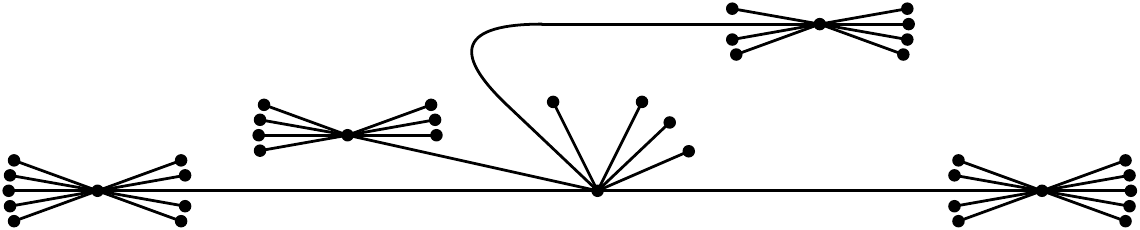}
    \caption{A bend is necessary in any drawing with height 5.}
    \label{fig:bent}
\end{figure}
\begin{figure}
    \centering
    \includegraphics{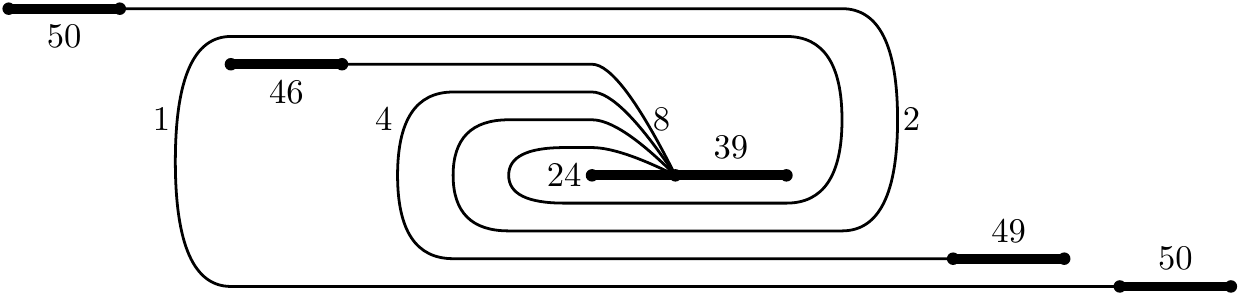}
    \caption{Spirals (e.g. the edge with weight 1)  may be necessary to draw weighted trees optimally.}
    \label{fig:spiral}
\end{figure}

The optimal height drawing of graphs is related to two significant classes of problems in computer science, and in particular, computational geometry and topology.
If, instead of ordered trees, we take (unordered) trees and allow edges to cross in the output drawing, we obtain the classical min-cut linear arrangement problem. This problem is well-studied \cite{Chu-etal85,Len82,Shi79} and Yannakakis~\cite{Yan85} presented an $O(n\log{n})$ time algorithm for drawing trees with optimal height in this sense. Of course, optimal drawings with straight-line edges always exist in this setting. On the other hand, it is known that the graph version as well as the weighted tree version~\cite{MoSu88} of the same problem is NP-hard. Since the trees corresponding to the reduction can be drawn optimally without self-intersection, it follows that optimal height drawing of unordered weighted trees is also NP-hard. All the mentioned problems lie in NP.

The optimal height graph drawing problem also shows up as a special case of an important open problem in computational geometry and topology called the homotopy height problem~\cite{Bie-etal19,Bur-etal17,Cha-etal18,ChLe09,Har-etal16}. 
In this context, a homotopy corresponds to a one-parameter family of curves $\gamma_i$ ($i\in[0,1]$) that sweeps a surface in a continuous way, where $\gamma_0$ and $\gamma_1$ are part of the input.
Roughly speaking, the homotopy height problem considers a surface homeomorphic to a sphere, disk, or annulus, endowed with a metric, and asks for a homotopy of curves that sweeps the surface in such a way that the longest curve $\gamma_i$ is as short as possible.
For a perfectly round sphere, the homotopy height is the length of its equator.
For the purpose of computation, discrete versions of the problem have been considered, where the surface is endowed with a cellular decomposition, and the lengths of curves are measured by the number of intersections with cell boundaries.
Each curve in general position with the cellular decomposition can be represented as a walk on the dual graph of the decomposition.
The vertices of the dual graph are represented geometrically as representative points of cells, and edges of the dual graph correspond to pieces of cell boundaries shared by two cells.
As a curve sweeps over the surface, it can sweep over vertices of the dual graph (resulting in a face flip), or create or remove pairs of intersections with edges of the embedded graph (resulting in a spike or unspike).
Figure~\ref{fig:dhomotopy} illustrates a dual graph (of a cellular decomposition) with vertices $P$ and $Q$, and a homotopy through curves $\gamma_i$ connecting $P$ to~$Q$. If the cell decomposition contains exactly one vertex, then its dual is a tree and the problem of homotopy height becomes equivalent to drawing trees with optimal height (In this case, the starting and ending curves are nested circles in the unbounded face so that the curves sweep an annulus).

\begin{figure}
    \centering
    \includegraphics{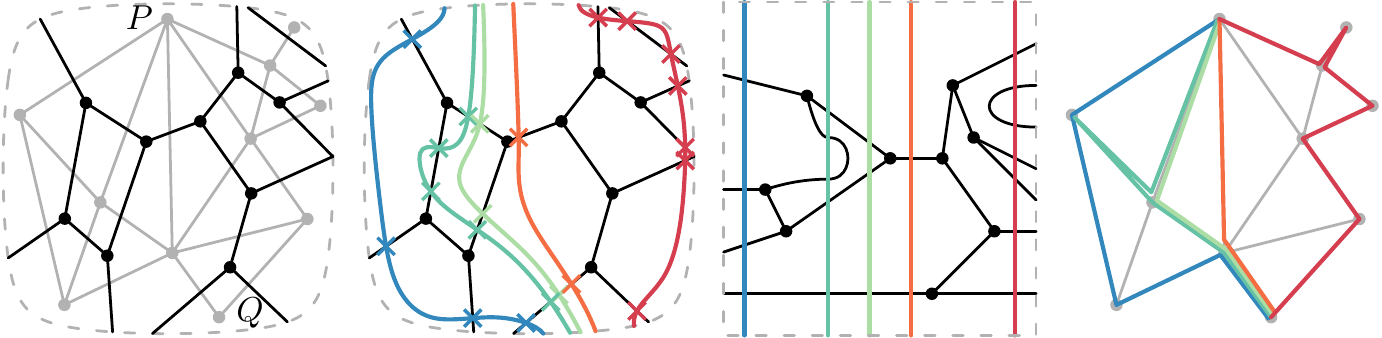}
    \caption{Left: a cellular decomposition of a disk (black) and its dual. Middle: some curves of a homotopy whose curves start at $P$ and end at $Q$, and a homeomorphism of the disk that sends the curves to vertical lines. Right: the corresponding walks in the dual graph.}
    \label{fig:dhomotopy}
\end{figure}

Although homotopy height admits an efficient $O(\log n)$-approximation algorithm~\cite{Har-etal16}, its exact computation appears to be very challenging.
In fact, it was only recently shown to lie in the complexity class NP~\cite{Cha-etal18} in the setting of edge-weighted graphs.
If the curves at the start and end of a homotopy are disjoint and shortest curves, it is known that there exists an optimal homotopy that sweeps the surface in a monotone fashion~\cite{CMOR17}, that is, any point of the surface is swept at most once.
Homotopy height is closely related to other important graph parameters~\cite{Bie-etal19}.

The duality relation between graph drawings and homotopy height is as follows. Consider a plane graph $G$. The homotopy height problem asks for a homotopy that starts as a closed curve around the graph and ends in a closed curve in the outer face that does not go around the graph. Such a homotopy must sweep over the entire graph $G$, and there exists an optimal homotopy that does this in a monotone way. Each curve of such a homotopy starts and ends in the outer face.
Drawing the intersection pattern of $\gamma_i$ with $G$ on the vertical line at $x=i$, results in a drawing of $G$ whose height is the homotopy height.
Conversely, any optimal drawing of $G$ can be turned into an optimal homotopy that is monotone. See Figure~\ref{fig:dhomotopy}.

In this paper instead of graphs we consider plane trees or ordered trees. We present the first polynomial-time algorithm for the optimal height drawing of unit weight plane trees. 
Our results give a polynomial algorithm for the homotopy height of unit-weight one-vertex (multi-)graphs. This might point to the possibility that the problem for the general graphs is also polynomial. However, already in our restricted setting, the algorithm is quite involved and does not have a clear extension to general graphs.

Although the term height has recently been used for the problem we consider, there exist related but different parameters of graph drawings that also quantify some notion of height~\cite{ALP2018, Bie-etal19, MAR2011}. 

\section{Background and terminology}

\subsection{Drawings and local disks}
    \subparagraph{Drawings.}
    Formally we work with \emph{plane trees} instead of ordered trees. This is just some reasonable, e.g. piecewise-linear, drawing $g\from T \to \mathbb{R}^2$ of a finite tree $T$ in the Euclidean plane. This plane drawing is fixed once and for all for any ordered tree $T$ and respects the given ordering around each vertex. In order to distinguish the Euclidean plane containing the drawing $g$ we use the symbol $\Pi$ for this plane, so that $g(T) \subset \Pi$.
    
    \subparagraph{Convention.} With a slight abuse of notation, we will not distinguish between $T$ and its embedding $g(T) \subset \Pi$ in the plane. We use the words edge and path for edges and simple curves on $T$ exclusively, and reserve the word curve for curves in the drawing plane (the plane to which $\Pi$ is mapped).
    
    \subparagraph*{}
    A \emph{drawing} $\phi$ of a tree $T$, is a continuous injective function mapping $\Pi$ into $\mathbb{R}^2$. We consider only drawings $\phi$ in which the image of every edge $e$ is piecewise-linear, and such that every vertical line intersects the drawing in a finite number of points. It is not difficult to see that this restriction does not affect the optimal height of the drawing.
    In our figures, for aesthetic purposes, we often draw edges as smooth curves.
    
    Let $E=E(T)$, $V=V(T)$ denote the set of edges and vertices of $T$. We always denote the number of vertices by $n$. By $H(\phi)$ we denote the \textit{height} of the drawing $\phi$. That is, the maximum number of points of the drawing on a vertical line.

    \subparagraph{Local disks.}
    Let $D\subset \Pi(T)$ be a topological disk in the plane in which $T$ is drawn. We denote the boundary of $D$ by $\partial D$. Let $T_D = T\cap D$ and assume $T_D$ is connected.
    We say that an edge $e \in E$ is a \emph{boundary edge} of $D$ if $e \cap \partial D \neq \emptyset$. We call an edge \emph{internal} if it lies in the interior of $D$. We denote by $B(D)$ the set of boundary edges of $D$. 
    Let $\partial D = C_l \cup C_r$ where $C_l \cap C_r = \{p_N,p_S\}$ is a set of two points, where none is in $T$. Intuitively, we think of $C_l$ and $C_r$ as the left and right boundary of $D$.
    This ``partition'' of $\partial D$ divides the set of boundary edges $B(D)$ into \emph{left} and \emph{right boundary edges} $B(D) = B_L(D) \sqcup B_R(D)$. We call $(D, B_L(D), B_R(D))$ a \emph{local disk}.
     
     
     A \emph{drawing of a local disk} $(D, B_L, B_R)$ is a homeomorphism $\phi\from D \to Q$ onto a rectangle $Q$ with edges $(\beta_l,\beta_t, \beta_r, \beta_d)$, such that under $\phi$, the boundary edges in $B_L$ intersect $\beta_l$, and those in $B_R$ intersect $\beta_r$ and $\phi(T_D) \cap (\beta_t \cup \beta_r)=\emptyset$. See Figure~\ref{fig:localdisk}. Note that we can select a local disk whose interior contains the whole tree, such that $T_D=T$ and there are no boundary edges. The height of the left (right) boundary in any drawing is the number of left (right) boundary edges of the local disk, and we call this number as the \textit{left (right) boundary height}. When the two boundary heights are equal we simply say boundary height.
     
    \begin{figure}
        \centering
        \includegraphics{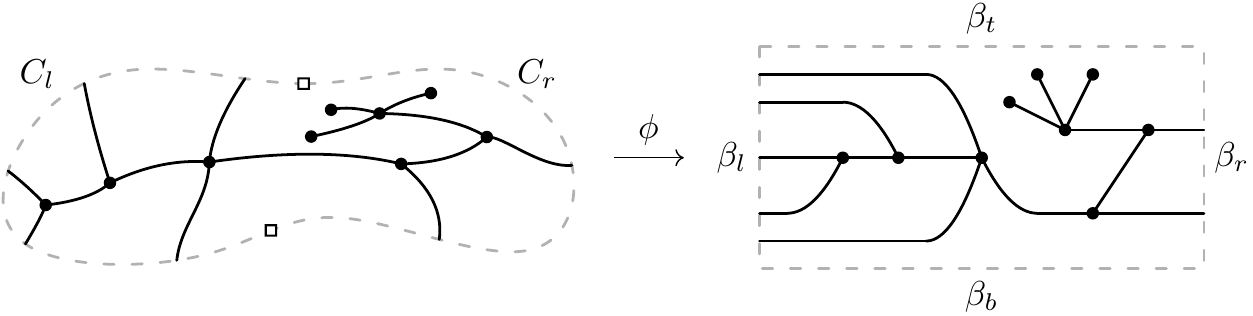}
        \caption{A local disk and a drawing.}
        \label{fig:localdisk}
    \end{figure}

    \subparagraph{The move sequence of a drawing.}
    Consider sweeping a vertical line over a drawing of $T$ (or the interior of a local disk).
    The sweep line encounters three types of events: left bends (points interior to edges of $T$ whose $x$-coordinate in the drawing is locally minimal), right bends (symmetric to left bends), and vertices.
    We will refer to these events as \textit{moves}, and the corresponding point of $T$ as its \textit{location} (i.e. the vertex corresponding to a vertex move, or the point interior to the edge corresponding to the bend move).
    We assume that all bends and vertex moves occur at distinct $x$-coordinates, and refer to the left-to-right sequence of moves of a drawing as its \emph{move sequence}.
     
\subsection{Cuts and shortcuts}
Let $D$ be a local disk. By a \emph{cut} in the local disk $(D,B_L,B_R)$ we mean the sequence of edges crossed by a curve that connects $p_N$ to $p_S$, where $p_N$ and $p_S$ are some two points giving rise to the local disk $(D,B_L,B_R)$ (an edge might repeat consecutively in the sequence). Some times we refer to the curve itself as a cut. Note that the same local disk can be defined with many such pair of points but this choice is not important. The \emph{length} (or \emph{height}) of a cut is the number of edges in it (counted with repetition), or the number of intersections of the curve with the tree $T_D$. A cut $C$ is a \emph{shortcut} if its length is smallest over all cuts of $D$. For the proof of the following lemma we refer to \cite[Lemma 4.2]{Cha-etal18}.

\begin{lemma}[Pausing at a shortcut]\label{l:pausing}
Let $D$ be a local disk, $\phi\from D \to Q$ a drawing and $C$ a shortcut in~$D$. There is a drawing $\phi'$ of height less than or equal to the height of $\phi$ in which there is a vertical line defining the cut $C$. Moreover, vertical lines of $\phi$ that are disjoint from $C$ are unaffected and appear in $\phi'$.
\end{lemma}

We say that the drawing $\phi$ can \emph{pause} at the shortcut $C$, resulting in the drawing $\phi'$. When a cut~$C$ is vertical in an optimal drawing and each sub-disk cut by~$C$ contains a connected part of $T_D$, then $C$ subdivides the problem into two sub-problems whose optimal drawings can easily be merged to form an optimal drawing of the original disk.

\section{Overview of the algorithm}

Our main result is an algorithm for computing optimal drawings of plane trees.
This algorithm is a dynamic program which in a high level works as follows.
Each cell of the dynamic programming table represents a local disk and stores the optimal height of that disk (or an optimal drawing, if an optimal drawing is to be computed). The local disks represented by the cells are of two special types: spine disks and skew spine disks (defined in Section~\ref{s:spinedisks}). These disks essentially are local disks that cannot be cut by shortcuts. Row $m$ of the table stores all spine or skew spine disks with exactly $m$ interior vertices. For $m>1$, row $m$ of the table is built using the information in lower rows in two phases. The first phase constructs all possible $m$-vertex spine and skew spine disks. The second phase computes the height of an optimal height drawing for each of the computed disks of row $m$ (or computes a drawing, if the the drawing is needed). The computed optimal height (or optimal drawing) will be stored again in the table. The base of the table consists of spine or skew spine disks with a single interior vertex. The possibilities for the decomposition of a single spine disk or skew spine disk into such disks with smaller number of vertices is shown in Figures~\ref{fig:spinedecom1} and~\ref{fig:spinedecom2}. With this description, a final optimal drawing consists of drawings in Figures~\ref{fig:spinedecom1} and~\ref{fig:spinedecom2} nested inside each other, and where the deepest level is a single vertex disk. A trapezoid (representing a skew spine disks) will fit into a trapezoid and a rectangle into a rectangle (spine disk).

There are two ingredients in the proof of correctness. First, we show in Proposition~\ref{l:bpdrawing} that any (sufficiently general) drawing can be turned, without increasing the height, into a drawing that has a hierarchical structure. The root of this structure tree is a spine disk containing the whole tree (with zero boundary edges). The nodes of the structure tree are (skew) spine disks. Each node is cut essentially into a collection of sub-disks, using shortcuts that are made vertical via pausing. These sub-disks are (skew) spine disks that form the children of the node. Each node, has one of polynomially many possibilities for the decomposition, depicted in Figures~\ref{fig:spinedecom1}~and~\ref{fig:spinedecom2}. The spine disks corresponding to leaves of any structure tree are single-vertex local disks and thus trivial to draw optimally. In brief, any drawing can be turned into one which has a tree structure of spine and skew spine disks without increasing the height. 



The second ingredient of the proof of correctness is a proof that there exists some optimal drawing such that a super-set of all the (skew) spine disks in its tree structure can be enumerated in polynomial time. For this purpose, we define a quality measure for a drawing. To rule out pathological drawings and simplify our arguments, we need to consider simplified drawings which are the result of applying simplification moves of Figure~\ref{fig:simplification}. We also consider balanced drawings, which are ones where the height of the lines on both sides of (and very close to) any vertex differ by at most one. Among all the optimal drawings, we take the drawing which is simplified, balanced, and maximizes our quality measure. Lemma~\ref{l:comdecom} asserts that such an optimized drawing has itself a tree structure of (skew) spine disks. The tree structure of a drawing which optimizes a slightly stronger measure, namely the secondary quality, is called a fat structure. Proposition~\ref{p:spinechar} characterizes the spine disks that can appear in a fat structure. This description allows us to easily enumerate all possible spine disks that can appear in a fat structure in polynomial time and only store the (skew) spine disks in our table that conform to this characterization. This will result in a polynomial-sized table and hence a polynomial algorithm. 
%
%

\section{Simplifying the drawings}
    
    
    \begin{figure}[b]
        \centering
        \includegraphics{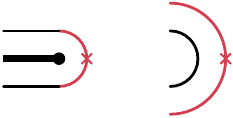}
        \caption{A right bend (marked) stuck around a vertex (left) or stuck around a bend (right). The bold line represents a bundle of arbitrarily many edges incident to the vertex.}
        \label{fig:stuck}
    \end{figure}
    Let $\phi$ be a drawing of a local disk $D$ and $T=T_D$.
    We label the left (resp. right) bends of $\phi$ as either \textit{stuck} or not, depending on whether the bend encloses the next (resp. previous) move.
    Figure~\ref{fig:stuck} illustrates the two possible reasons for a right bend to be stuck.
    Two consecutive moves of a drawing may admit a simplification (of the drawing) that replaces the two sequence by a simpler sequence of moves without increasing the height of the drawing.
    For each of these simplifications, either the first move is a non-stuck left bend, or the second move is a non-stuck right bend.
    We explain the types of simplifications involving a non-stuck right bend (see Figure~\ref{fig:simplification}), the types involving a non-stuck left bend are symmetric.
    As mentioned, the second move is a non-stuck right bend, so we distinguish cases based on the first move of the pair.
    \begin{figure}
        \centering
        \includegraphics{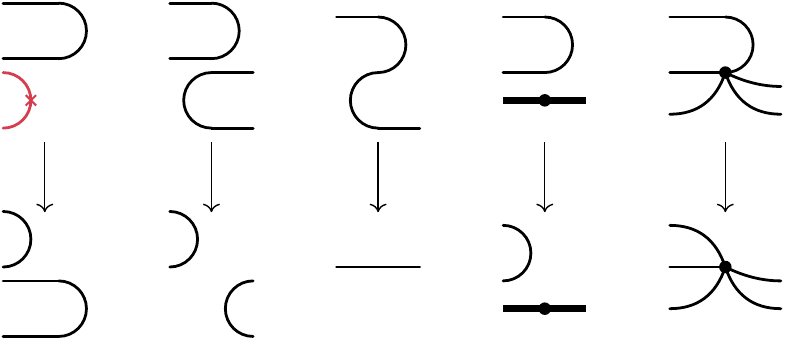}
        \caption{Left to right: stuck slide, bend-bend separation, bend-bend cancellation, vertex-bend separation, (strong) vertex-bend cancellation.
        }
        \label{fig:simplification}
    \end{figure}
    \begin{description}
        \item[Stuck slide.]
            In this case, the first move is a stuck right bend.
            The non-stuck right bend does not enclose the stuck right bend (otherwise it would also be stuck).
            Exchanging the order of the two bends ensures that neither of the resulting bends are stuck. 
        \item[Bend-bend (resp. vertex-bend) separation.]
            The first move is a left bend (resp. vertex) that is not connected to the right bend. We exchange the moves, reducing the height of the line in between.
        \item[Bend-bend cancellation.]
            The first move is a left bend that is connected to the right bend. We replace the bends by an $x$-monotone curve, reducing the number of bends.
        \item[Vertex-bend cancellation.]
            The first move is a vertex that is connected to the right bend. We replace the bend by an $x$-monotone curve, reducing the number of bends.
            We call a vertex-bend cancellation \emph{strong} if the simplification does not decrease the absolute difference between the number of edges incident to the left and right of the vertex.
    \end{description}
    We say that a drawing $\phi$ is \emph{strongly simplified} if no simplification move is possible, and \emph{simplified} if only strong simplification moves are possible.
    
    We say that $\phi$ is \emph{balanced} if for any vertex~$v$, the heights of the vertical lines immediately to the left and right of $v$ are equal if the degree of $v$ is even, and differ by 1 if the degree of $v$ is odd.
    Balanced drawings will be useful for our algorithms.
    However, strong vertex-bend cancellations may make vertices less balanced.
    
    \begin{lemma}\label{l:simpli}
        If there is a drawing $\phi$ of height $H$ of a local disk $D$, then there exists a balanced simplified drawing of $D$ of height at most $H$ with a bounded number of moves.
    \end{lemma}

    \begin{proof}
    We call two cuts of $D$ equivalent if their overlay with the drawing defines the same combinatorial map.
    There are only a bounded number of equivalence classes of cuts of $D$ of length at most $H$, and each vertical line between moves of $\phi$ corresponds to one of these cuts.
    Over the course of this proof, we will modify $\phi$.
    If there exist two vertical lines that are equivalent and separated by at least one move, then we remove all the moves in between without increasing the height of $\phi$.
    After removing the moves between all such pairs of vertical lines, we are left with a drawing of $D$ of height at most $H$ and a bounded number of moves.
    
    We show that the (strong) simplification moves eventually terminate.
    No simplification introduces a new bend, and cancellation moves reduce the number of bends, so there are only a bounded number of cancellation moves.
    Let $R=\{i\mid\text{the $i$-th move is a right bend}\}$ be the set of indices of moves that are right bends, and $L=\{i\mid\text{the $(m-i)$-th move is a left bend}\}$ (where $m$ is the number of moves in $\phi$) be the set of indices of left bends (counted from the last move).
    Separation moves decrease at least one value of $L$ or $R$ while keeping other values the same, and no move increases any value of $L$ or $R$.
    Because the values in $L$ and $R$ are bounded non-negative integers, there are a bounded number of separation moves.
    Finally, we show that there are a bounded number of stuck slides.
    If there were an unbounded number of stuck slides, then because the number of other moves is bounded, there exists an unbounded sequence of stuck slides without other moves in between.
    Stuck slides do not affect the total (already bounded) number of moves $m$.
    A stuck slide involves a stuck right (resp. left) bend at position $i$ and a non-stuck right (resp. left) bend at position $i+1$ (resp. $i-1$).
    The number of stuck right (resp. left) bends at positions $\leq i$ (resp. $\geq i$) strictly decreases and the set of stuck left (resp. right) bends stays the same.
    As such, every stuck slide strictly decreases the nonnegative integer potential function $\sum_{i\in L'}i^2+\sum_{i\in R'}i^2$, where $R'=\{i\mid\text{the $(m-i)$-th move is a stuck right bend}\}$ and $L'=\{i\mid\text{the $i$-th move is a stuck left bend}\}$, so the (strong) simplification terminates.
    
    Now that $\phi$ is strongly simplified, we turn it into a simplified balanced one.
    We essentially apply strong vertex-bend cancellation in reverse for every vertex that is not balanced, see Figure~\ref{fig:balance}.
    For this, we need to be careful not to introduce simplification moves other than strong vertex-bend cancellations.
    Consider a vertex that is not balanced, and assume that it has $k$ edges incident to its left side and $\leq k-2$ edges incident to its right (the case where there are more edges on the right is symmetric).
    We perform the following balancing move.
    Pick the topmost edge incident to the left and bend it over the top of the vertex, decreasing the number of edges incident to the left by one and increasing the number of edges incident to the right by one.
    Note that a balancing move does not increase the height of the drawing, so repeatedly applying balancing moves, we obtain a balanced drawing whose height does not exceed $H$.
    \begin{figure}
        \centering
        \includegraphics{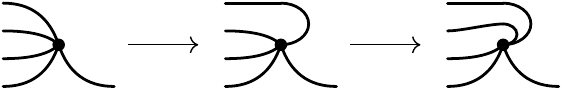}
        \caption{Balancing a vertex.}
        \label{fig:balance}
    \end{figure}
    
    It remains to show that the resulting drawing is simplified.
    For this, consider a balancing move and assume that no simplification move is possible before the balancing move, and suppose for a contradiction that a simplification move is possible after the balancing move.
    The simplification move involves either the vertex or the bend of the balancing move, but not both, because that would require a strong simplification move.
    Without loss of generality, assume that the balancing creates a right bend.
    If the simplification involves the vertex, then it involves a left bend left of the vertex.
    Before the balancing move, if this left bend was not connected to any of the $x$-monotone curves incident to the left of the vertex, then a vertex-bend separation was already possible.
    So the left bend is connected to (because our graph is a tree) exactly one of the $x$-monotone curves incident to the left of the vertex, and does not enclose any of the $x$-monotone curves (because otherwise, there would need to be an other move inside the bend), but then this bend could have been cancelled against the vertex.
    So the simplification involves the created right bend, and a right bend next to it instead.
    Therefore, it is a stuck slide, and the created bend must be stuck, but this bend supports a strong cancellation, so it is not stuck, which is a contradiction.

    \end{proof}

    \begin{lemma}\label{l:finiteness}
    Any simplified drawing of height $H$ of a local disk $D$ with $n$ vertices has at most $(H+1)n$ moves if $n>0$, and at most $H$ moves if $n=0$.
    \end{lemma}

    \begin{proof}
        If there exists a left bend immediately followed by a right bend, then those can be cancelled or separated, so any simplified drawing has the following pattern of moves.
        \begin{enumerate}
            \item Before the leftmost vertex, there are only left bends.
            \item After the rightmost vertex, there are only right bends.
            \item For any consecutive pair of vertices (i.e. there are no other vertices in between their x-coordinates) there is a (possibly empty) sequence of right bends followed by a (possibly empty) sequence of left bends.
        \end{enumerate}
        If there is a sequence of more than $H/2$ consecutive left (or right) bends, then the height exceeds~$H$, so any simplified drawing has at most $(2H/2+1)n=(H+1)n$ moves.
    \end{proof}

    \begin{observation}
    Let $\phi$ be a balanced drawing of the local disk $D$. Then applying all possible non-strong simplifying moves to $h$ keeps the drawing balanced.
    \end{observation}

\section{Bubbling a sub-tree}

    Let $e$ be an edge of a tree $T$. There are two sub-trees $T_1$, $T_2$ of $T$ that result from removing $e$. For $i\in \{ 1,2 \}$, we call the rooted trees $T_i=T_i(e)$, with the root chosen to be the endpoint of $e$ in $T_i$, the \emph{rooted sub-trees anchored via the edge $e$} and the edge $e$ the \emph{anchor edge} of the rooted tree $T_i$. We call the endpoint of $e$ which is not the root of $T_i$ the \emph{anchor vertex} of $T_i$. The \emph{exposed height} of the sub-tree $T_i \cup \{e \}$, denoted $eH(T_i,e)$, is the height of the optimal height drawing of a local disk containing $T_i$ in its interior and such that the anchor edge $e$ is the single boundary edge, see Figure~\ref{fig:exposed}. We call such a drawing of $T_i$ an \emph{exposed drawing} of the sub-tree $T_i$ with respect to $e$.
    \begin{figure}
        \centering
        \includegraphics{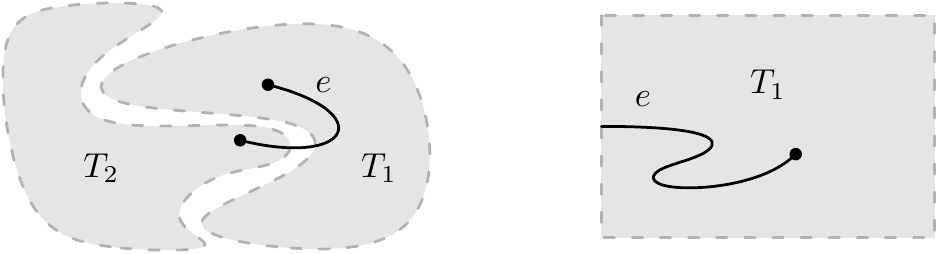}
        \caption{Two local disks (containing sub-trees $T_1$ and $T_2$) with a single boundary edge $e$ (left). An exposed drawing of the sub-tree $T_1$ with anchor edge $e$ (right).}
        \label{fig:exposed}
    \end{figure}
    
    Let $P$ be a simple path (possibly of length\footnote{The length of a path is the number of its edges.} 0) in $T$. The \emph{neighborhood} of $P$, denoted $N(P)$, is the subset of vertices of $T$ not in $P$ which are connected in $T$ to some vertex of $P$ by an edge. We say that a rooted sub-tree $T'' \subset T$ is a sub-tree \emph{anchored at $P$} if $V(T'')\subset V(T)-V(P)$ and the root of $T''$ is in $N(P)$. It  follows that, for any $P$, the edge-set of $T$ is partitioned into three sets, the edges of $P$, the edges of sub-trees anchored at $P$, and the anchor edges incident to $P$.
    
    We call an exposed drawing of a sub-tree $T'$ in a drawing of $T$ \textit{a bubble} if the strip of the plane containing this exposed drawing contains no moves of the rest of the drawing. In other words, we can compress the exposed drawing of $T'$ into a drawing inside an arbitrary small bubble, without affecting the height of the drawing of $T$.
    
    Let $\phi$ be a drawing of the tree $T$ and let $T' \subset T$ be a sub-tree anchored using an edge $e$ to a vertex $v$. We say that $\phi'$ is obtained by \emph{bubbling the sub-tree $T'$ at the point $t \in \mathbb{R}^2$} if $\phi'$ is such that i) $H(\phi')\leq H(\phi)$, ii) $T'$ is drawn in a bubble with $e$ the boundary edge and $t$ being the point of $e$ on the boundary, iii) the drawing is changed only over $T'$ and the edge $e$, and iv) the $\phi'$-image of $e$ is contained in $\phi$-image of $T' \cup e$. See Figure~\ref{fig:bubbling} for an example, where $t=t_r$. One of our main observations is that bubbling is always possible at a suitable $t$.
    \begin{figure}[b]
        \centering
        \includegraphics{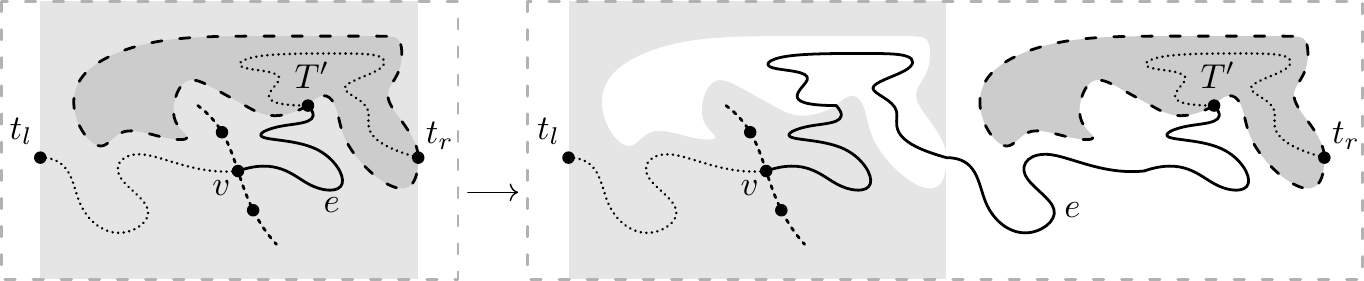}
        \caption{Bubbling the sub-tree $T'$, $t_l$ and $t_r$ are locations of extreme moves in the given drawing of the local disk}
        \label{fig:bubbling}
    \end{figure}
    
    \begin{lemma}\label{l:bubble}
    Let $T$ be a tree and $h$ be a drawing of $T$. Let $T'$ a sub-tree anchored at a vertex $v$. Let $t_l$ and $t_r$ be the points of $T$ which have the smallest and the largest $x$-coordinate, respectively. We can assume these points are unique. If $T'$ contains exactly one of $t_l$ and $t_r$, then $T'$ can be bubbled at that point. 
    \end{lemma}

    \begin{proof}
    Without loss of generality, assume that $t_l$ is not in the sub-tree $T'$, see Figure~\ref{fig:bubbling}. Let~$e$ be the edge connecting the root of $T'$ to a vertex $v$.  Let $Q$ be the (unique) path connecting $v$ to $t_r$, and let $S$ be the path from $t_l$ to the root of $T'$. We take an arbitrary point $s \in e$ and break the tree $T$ into $T_1$ and $T_2$, where $T_2$ contains the sub-tree $T'$ and $s$ is a vertex in $T_1$ and $T_2$ and $T_1 \cap T_2 = \{s\}$. Let $e_1$ and $e_2$ be the edges in $T_1$ and $T_2$, respectively, which are pieces of $e$. Let $\phi_1$ be the drawing of $T_1$ obtained by restricting $\phi$ to $T_1 - e_1$ and setting the image of $e_1$ in $\phi_1$ to be the image of the path $Q$ under $\phi$. Observe that the height of $\phi_1$ is at most $H(\phi)$ and the image of $s$ lies at the largest $x$-coordinate in $\phi_1$. Define $\phi_2$ by restricting $\phi$ to $T'$ and define the image of $e_2$ to be the image of the path $S$. Again, the height of $\phi_2$ is at most $H$ and the image of $s$ lies at the smallest $x$-coordinate. 
    Now we can compress the drawing $\phi_2$ in the $x$-direction to make the $x$-interval of $T'$ as small as necessary. The possibility of making the bubble small implies we can make sure that there are no other moves over the $x$-coordinate of the bubble, in case we are doing this operation on part of a larger drawing. The drawing required by the lemma is then the drawing $\phi'$ which is obtained by concatenating $\phi_1$ and $\phi_2$, possibly after translating one of them in the $y$-direction.
    Clearly, $\phi'$ has height at most $H(\phi)$.
    \end{proof}

    
\section{Spine disks}\label{s:spinedisks}
Let $D$ be a local disk and $B(D)$ the set of its boundary edges. Recall that the boundary edges of a local disk are divided into left boundary edges, $B_L(D)$, and right boundary edges, $B_R(D)$. Let $e_l \in B_L(D)$ and $e_r \in B_R(D)$. We say that $e_l$ and $e_r$ are \emph{opposite} one another if they are incident to the same vertex in the interior of $D$.
    We call a local disk a \emph{spine disk} with \emph{spine} $P$, if all of the following hold:
    
    \begin{enumerate}
        \item\label{i:spine1} $P$ is a simple path in $T_D$, such that every boundary edge is incident to a vertex of $P$, and $P$ is interior-disjoint from boundary edges of $T_D$.
        
        \item There is a bijection $\alpha\from B_L(D) \to B_R(D)$ such that each $e$ is opposite $\alpha(e)$.
        
        \item If $P$ has at least two vertices, there are boundary edges incident to its extremal vertices.
        
    \end{enumerate}
    
    If there are no boundary edges, we let $P$ be an arbitrary vertex, so that $P$ is always defined. Therefore a local disk that contains all of the input tree $T$ and $P$ chosen to be any vertex is a spine disk.
    
    A \emph{skew spine disk} is a spine disk to which a new boundary edge incident to some vertex of $P$ is added. It follows that the height of one boundary line of any drawing of a skew spine disk is one more than the height of the other boundary line. We call a (skew) spine disk a \textit{vertex disk} if there is a single vertex in its interior. Figure~\ref{fig:spdef} shows an optimal drawing of a tree and a ``decomposition'' of the drawing into spine (rectangle) and skew spine disks (trapezoids). Note that a skew spine disk with a single boundary edge is a bubble. All skew spine disks in Figure~\ref{fig:spdef} are bubbles.
    
    \begin{figure}[b]
        \centering
        \includegraphics{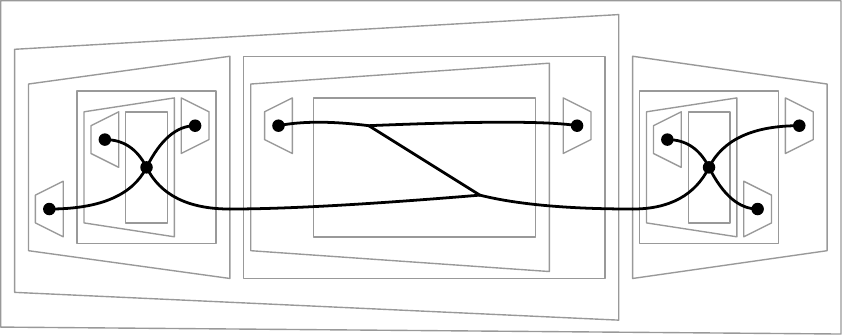}
        \caption{Spine and skew spine disks.}
        \label{fig:spdef}
    \end{figure}

\subsection{Spine decomposition}
    
    We introduce some terminology before stating one of our main propositions. Let $D$ be a local disk and let $C$ be a collection of disjoint shortcuts (combinatorially distinct from the left and the right boundary lines) in $D$, and let $C$ cut the disk $D$ into disks $D_1, \ldots, D_m$. According to Lemma~\ref{l:pausing}, an optimal drawing $\phi$ of $D$ can be obtained by gluing optimal drawings $\phi_i$ of the $D_i$, $i=1,\ldots,m$. Then we say $\phi$ is obtained by \emph{merging} the drawings $\phi_1, \ldots, \phi_m$. 
    In our schematics, we draw a rectangle for a spine disk and a trapezoid for a skew disk. The shorter side of the trapezoid has one less boundary edge than the long side. A vertex inside a rectangle or a trapezoid indicates a vertex disk. A thick black line is a collection of parallel lines whose number is indicated. A \emph{pipe} in a drawing bounded by two vertical lines $l$ and $r$ is a subpath of an edge of $T$ drawn as an $x$-monotone curve between these lines. Observe that a pipe can always be drawn as a straight line connecting the lines $l$ and $r$.

    If $D$ is a vertex spine disk or vertex skew spine disk then there is a trivial, straight-line, optimal drawing of $D$. These disks form the building blocks of our drawings. The following proposition shows how more complicated (skew) spine disks can be decomposed into less complicated ones and eventually into vertex (spine) disks.


    \begin{proposition}[Spine Decomposition] \label{l:bpdrawing}
    Let $D$ be a spine (resp. skew spine) local disk with $b\geq 0$ boundary edges on one side and $b$ (resp. $b+1$) boundary edges on the other side.
    If $D$ is not a vertex disk and $D$ has a drawing of height $H$, then $D$ has a drawing of height at most $H$ that can be decomposed as one of the cases of Figure~\ref{fig:spinedecom1} (resp.~\ref{fig:spinedecom2}), up to horizontal and vertical reflection. In these drawings $m, a_i, c_j$ are non-negative integers. 
    \end{proposition}

    \begin{figure}
        \centering
        \includegraphics[scale=1.1]{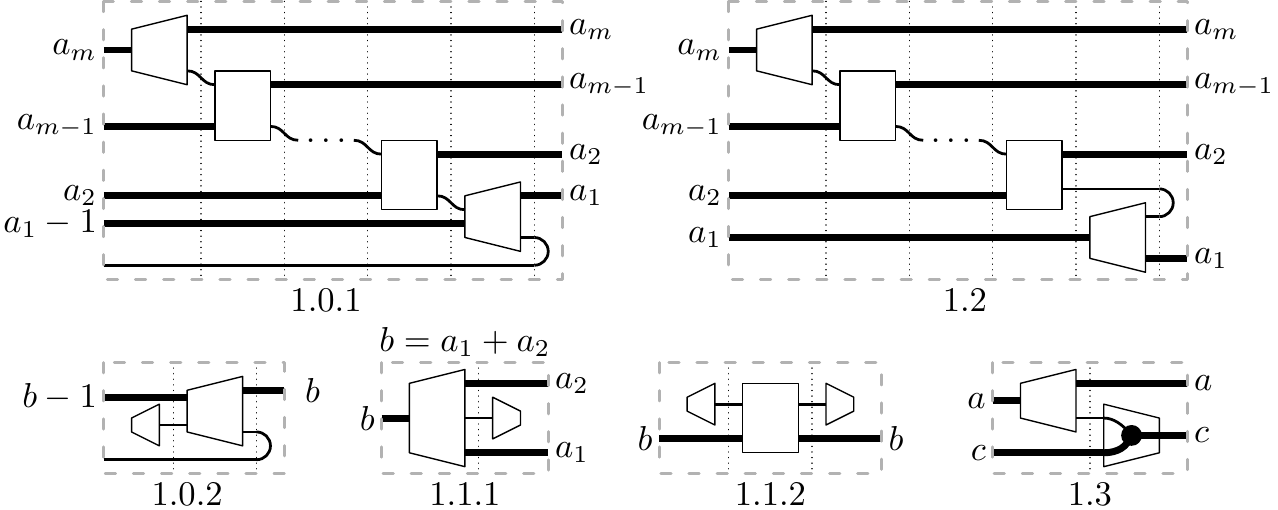}
        \caption{Decomposition of spine disks. Thick lines indicate bundles of parallel edges. The number of parallel edges in bundles are indicated by the labels on the sides. The values $a_i$ can be $0$. Rectangles indicate spine disks and trapezoids indicate skew spine disks. A black dot indicates a vertex disk.}
        \label{fig:spinedecom1}
    \end{figure}
    \begin{figure}
        \centering
        \includegraphics[scale=1.1]{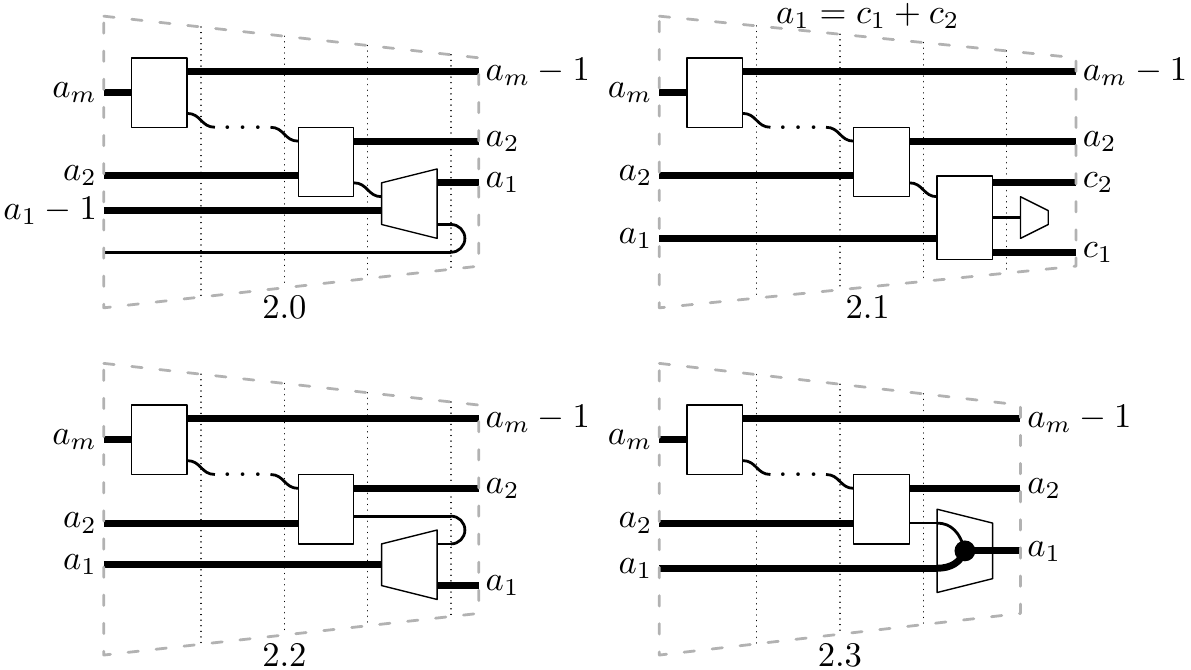}
        \caption{Decomposition of skew spine disks.}
        \label{fig:spinedecom2}
    \end{figure}
 
\begin{proof}
        In the case of skew spines, we consider only the case where the left boundary has height $b+1$ (the other case is symmetric). Using Lemma~\ref{l:simpli} we can assume that the drawing $h$ is simplified and balanced.
        Let $t_l$ and $t_r\in T_D$ be the points corresponding to the leftmost and rightmost move inside $D$ in $h$, respectively.
        Because in a (skew) spine disk, any path that starts and ends on the same (left or right) boundary contains at least one vertex, $t_l$ will be either a vertex or a left bend, and $t_r$ will be either a vertex or a right bend. For any side of the disk for which $b=0$, the corresponding first move can only be a vertex move. For instance, any drawing of the whole tree starts and ends with a veretx move.

        If $t_l=t_r$ then the (skew) disk is a vertex disk and the drawing is monotone and $t_l$ and $t_r$ equal the single vertex move. We assume therefore $t_l \neq t_r$.
        
        The argument now breaks into multiple cases. Case 0) is common to both spine and skew spine disks. 
        
        \subparagraph{Case 0) $t_r$ or $t_l$ is in the interior of a boundary edge incident to an internal vertex of $P$.}
        Assume $t_r$ lies in the interior of a boundary edge $e$. If $e$ is a right boundary edge, then there is a bigon that is bounded by $e$ and the line $x=x(t_r)$. This bigon has to be empty. There is therefore a left bend on the boundary of the bigon. This is a non stuck left bend and there is a simplifying vertex-bend separation or a bend-bend separation or a bend-bend cancellation would be possible. The case depends on what type of move is to the immediate right of the left bend. This contradicts the fact that the drawing is simplified. So we can assume that $t_r$ and (symmetrically) $t_l$ are not on right and left boundary edges, respectively.

        
        Now assume $e$ is a left boundary edge. We claim that $e$ is incident to one of the extreme vertices of $P$. See Figure~\ref{fig:case0}. Indeed, there is a disk $B$ bounded by the line $x=x(t_r)$ and part of $e$ and a right boundary edge opposite $e$ (or by a path in $P$ and then a right boundary edge, if $e$ has no opposite edge, in case of a skew disk). By the third condition of a spine disk, $B$ cannot contain any vertex of $P$ in its interior. Therefore, the endpoint of $e$ on $P$ cannot be an interior vertex of the path $P$. Thus $e$ is incident to an extreme vertex of $P$ and we are not in the current case.

        \begin{figure}[ht]
            \centering
            \includegraphics{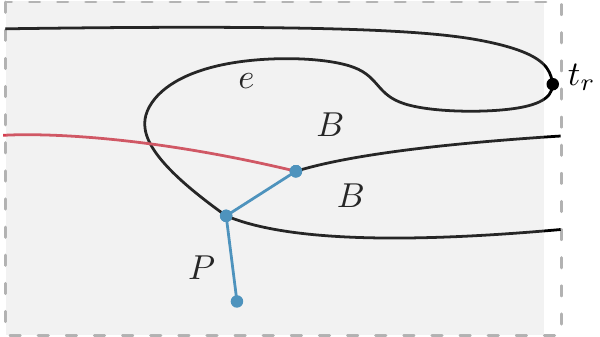}
            \caption{If the bigon $B$ contains a vertex of the spine then there will be an intersection. The red edge creates an intersection. Blue lines denote the boundary of the spine disk. Grey color shows the area where the drawing is unknown (other than some parts which are drawn).}
            \label{fig:case0}
        \end{figure}
        
        \subparagraph{Case 1) $D$ is a spine disk.} 

        Denote by $T_r$ and $T_l$ be sub-trees anchored on $P$ containing $t_r$ and $t_l$ (whenever defined), or anchored by an edge containing $t_r$ and $t_l$, respectively. Let $v_l$ and $v_r$ be vertices to which $T_l$ and $T_r$ are anchored. If $t_l$ is on $P$ and is a vertex let $v_l=t_l$, and if $t_l$ is in the interior of an edge let $v_l$ be the endpoint of that edge which is not on the path from $t_l$ to $v_r$. Similarly define $v_r$ in these cases. Observe that, in any case, the height immediately to the right of $t_l$ and immediately to the left of $t_r$ is at least $b+1$. Therefore, by Lemma~\ref{l:pausing}, we can divide the disk at any cut of length $b+1$ between these two lines.
        
        \subparagraph{Case 1.0) $t_r$ or $t_l$ is in the interior of a boundary edge.}
        Assume, after a possible reflection across the $y$-axis that $t_r$ is in the interior to a boundary edge. Let $e$ be the boundary edge containing $t_r$. We know that $e$ is incident to an extreme vertex of $P$. By the argument in Case 0), $t_l$ is not on the same edge as $t_r$. 
        
        \subparagraph{Case 1.0.1) $v_l$ is not an endpoint of $e$.}
         There exists a cut $\lambda$ of length $b+1$ that has, the bend around $t_r$ and the segment of $e$ from $t_r$ to the left boundary, to its right and the rest of the drawing to its left, see Figure~\ref{fig:spineexamples}. In addition, there is a sequence of cuts of length $b+1$, which together with $\lambda$ divide the disk $D$ into a bend, a skew disk and a set of spine disks as in Figure~\ref{fig:spinedecom1} Case 1.0.1), or a reflection of this schematic in the $x$-axis. 
         
        \subparagraph{Case 1.0.2) $v_l$ is an endpoint of $e$.}
         If $T_l$ is defined, we can bubble the sub-tree containing $T_l$ around $t_l$ without increasing the height of the drawing (see the next case for the details of this procedure). We then divide the disk using the line that separates the bubble of $T_l$ and the line $x=x(t_r)-\epsilon$ into two disks and a bend, see Figure~\ref{fig:spineexamples}. This gives us the schematic in Figure~\ref{fig:spinedecom1}, Case 1.0.2, or one that is obtained by a reflection along the $x$-axis. 
         If $T_l$ is not defined, that is, if $t_l$ is on $P$, then it has to be that $t_l=v_l$.
         We are in Case 1.3) after a reflection across the $y$-axis. So we do analogously to that case. 
         
         From this point on, we assume that $t_r$ and $t_l$ are not in an interior of a boundary edge.
    
        \subparagraph{Case 1.1) $t_r$ and $t_l$ are not on $P$.}
    
        \subparagraph{Case 1.1.1) $T_l \neq T_r$.}
        Let $F$ be the tree that is the union of the spine $P$, the sub-trees $T_{l}$ and $T_{r}$ and the two edges from $v_l$ and $v_r$ to the roots of $T_l$ and $T_r$, respectively. We apply Lemma~\ref{l:bubble} to the drawing of the tree $F$ to bubble the sub-tree $T_r$ around $t_r$. Let $\mathring{h}_F$ be the resulting drawing. From properties i), iii) and iv) of the bubbling procedure follows that if we replace $h|_F$ with $\mathring{h}_F$, the resulting drawing is well-defined. Moreover, observe that there exist $b$ edge-disjoint paths between the boundaries of the disk $D$ which are disjoint from $F$. It follows that the height of $\mathring{h}_F$ is at most $H-b$. Therefore, the height of the new drawing is at most $H$. Consider the line $\lambda$ of height $b+1$ that cuts out the bubble.
    
        We divide the disk $D$ into two sub-disks using $\lambda$. The left sub-disk is a skew disk, the right sub-disk contains a bubble. Thus we have Case 1.1.1) in Figure~\ref{fig:spinedecom1}.

        \subparagraph{Case 1.1.2) $T_l = T_r$.}
        Let $v:=v_l=v_r$ and $S:=T_l=T_r$ be the common sub-tree containing $t_l$ and $t_r$. We claim that $v$ is an endpoint of the path $P$. Otherwise, the third property of the spine disk implies that the image of $S$ lies in a sub-disk of $D$ bounded by edges of $P$, two left or two right boundary edges and the left or the right boundary of $D$. The two extreme points in the drawing of $D$  (i.e. $t_l$ and $t_r$) cannot be in the interior of this sub-disk. Thus our claim is proved.
        
        Let $P_l$ be the path from $v$ to $t_l$ and define $P_r$ analogously. Consider the the maximal path $P'$ contained in $P_l \cap P_r$. By assumption, the length of $P'$ is at least 1. Let $Q=P \cup P'$. Let $T'_r$ and $T'_l$ be the sub-trees anchored at $Q$ containing $t_r$ and $t_l$ in them, or in their anchor edge, respectively. Now observe that we can use Lemma~\ref{l:bubble} to bubble $T'_r$ and $T'_l$ just as in the last case. We then move the two boundaries over the two bubbles of $T'_r$ and $T'_l$. This gives Figure~\ref{fig:spinedecom1}, Case 1.1.2) or a figure obtained via a reflection along the $x$-axis. Note that the inner spine disk now has a longer spine path.

        \subparagraph{Case 1.2) $t_r$ is on $P$ and is not a vertex.}
        In this case, $t_r$ defines a bend in an edge $e$ of $P$. The line $x=x(t_r)-\epsilon$, for sufficiently small $\epsilon >0$, has weight $b+2$. Consider again $t_l$. Let $\hat{D}$ be the disk between the lines $x= x(t_r)+\epsilon$ and $x=x(t_r)-\epsilon$.

        If we remove the point $t_r$ the tree in $D$ (including the boundary edges) is disconnected into two parts. One of these parts contains $t_l$. We can then find a cut in $\hat{D}$ that cuts the left boundary edges of the component not containing $t_l$ and the right boundary edges of the other component. See Figure~\ref{fig:spineexamples}. This cut lies in the interior of $\hat{D}$ and has length $b+1$. Therefore, it is a shortcut. There exist now a sequence of cuts that can be used to divide $\hat{D}$ and thus $D$ into two skew spine disks and a sequence of spine disks as in Figure~\ref{fig:spinedecom1}, Case 1.2), or a similar schematic obtained via a reflection along the $x$-axis.
        
        \subparagraph{Case 1.3) $t_r$ is an extreme vertex of $P$.}
        Consider the line $\lambda$ with equation $x=x(t_r)-\epsilon$. This line has height at least $b+1$. If its height is $b+1$ then we are in the situation of Figure~\ref{fig:spinedecom1}, Case 1.3) or a similar schematic where the vertex disk is on top. Otherwise, the height of $\lambda$ is at least $b+2$. This contradicts the assumption the input drawing is balanced.
        
        
        \subparagraph{Case 1.4) $t_r$ is an internal vertex of $P$.}
        Consider the line $\lambda$ with equation $x=x(t_r)-\epsilon$. This line has height at least $b+2$. We can take one of the incident edges of $P$ on $t_r$ from left and bend it to be incident to $t_r$ from right, while not increasing the maximum height. This contradicts the assumption that the drawing is balanced. 
    
        \subparagraph{Case 2) $D$ is a skew spine disk.}
        Recall that a skew spine disk has a vertex on the spine $P$ where the incident left boundary edges are more than incident right boundary edges by one. We denote that vertex by $w$.

        \subparagraph{Case 2.0) $t_r$ is on a boundary edge.}
        By the arguments of Case 1.0), $t_r$ is on a boundary edges that is incident to a extreme vertex $v$ of the path $P$. Consider the line $x=x(t_r)-\epsilon$ and let $\hat{D}$ be the disk between the left boundary of $D$ and this line. Every cut of length $b+1$ in $\hat{D}$ is a shortest cut and can be used to cut the disk $\hat{D}$. There is one such shortcut for every edge on the path from $v$ to $w$, see Figure~\ref{fig:spineexamples}. Using these cuts, we obtain the schematic in Figure~\ref{fig:spinedecom2} Case 2.0) or a reflection of it. If $v=w$, we obtain a schematic that is obtained from this one, by deleting all the rectangles and connecting edges, other than the edge that connect the inner skew disk to a rectangle, this edge will be added to the left boundary. Namely, the schematic would contain a bend and a skew spine disk oppositely oriented. 
        
        From this point on, we assume that $t_r$ is not in an interior of a boundary edge.

        \subparagraph{Case 2.1) $t_r$ is not in $P$.} 
        Let $T_r$ be the sub-tree anchored at $P$ containing $t_r$ or which $t_r$ is in its anchor edge. We can think of the point where the extra left boundary intersects the left boundary as $t_r$. Then we can apply Lemma~\ref{l:bubble} to bubble the sub-tree $T_r$ around $t_r$. Consider the line $\lambda$ of height $b+1$ that separates the bubble in the new drawing. Let $\hat{D}$ be the disk between the left boundary of $D$ and $\lambda$. Let $v$ be the vertex to which $T_r$ is anchored. Also, let $m$ be the number of vertices on the path $v \cdots w$. If $m>1$, there are $m-1$ short cuts of length $b+1$ in $\hat{D}$, see Figure~\ref{fig:spineexamples}. These plus $\lambda$ divide the disk $D$ in a sequence of disks depicted in Figure~\ref{fig:spinedecom2} Case 2.1) up to a reflection along the $x$-axis. Other than two extreme spine disks in this schematic, other spine disks contain a single vertex in their spines. If $m=1$, we simply use $\lambda$ to cut the disk into a bend and a rectangle. The schematic is the same as before with $m=1$. The edge that connect the lowest rectangle to the the higher rectangle now will be added to the $a_1$ left boundary edges.
        
        \subparagraph{Case 2.2) $t_r$ is on $P$ and is not a vertex.}
        Then $t_r$ is a bend on $P$ and the line $\lambda$ with the equation $x=x(t_r)-\epsilon$ has height $b+2$. It follows that we can divide the disk $\hat{D}$, between the left boundary of $D$ and $\lambda$, by any cut of length $b+1$. Let $v$ be the first vertex that lies on the path from $t_r$ to $w$ and let $m$ vertices be on this path. There is a cut of length $b+1$ in $\hat{D}$ which cuts the edge containing $t_r$ once, as in Figure~\ref{fig:spineexamples}. There exist $m-1$ other cuts of length $b+1$, each cutting a distinct edge on the path from $v$ to $w$. We can cut the disk $\hat{D}$ into the schematic depicted in Figure~\ref{fig:spinedecom2} using these cuts, or a symmetric one obtained via a reflection across $x$-axis. If $m=1$, we do similar to the previous two cases.
        
        \subparagraph{Case 2.3) $t_r$ is an extreme vertex of $P$.}
        Consider the line $\lambda$ with equation $x=x(t_r)-\epsilon$. This line has height at least $b+1$. If its height is $b+1$ then we are in a situation very similar to the above cases and we obtain the schematic in Figure~\ref{fig:spinedecom2} or a symmetric one. If the height of $\lambda$ is at least $b+2$ we again reach a contradiction.
        
        \subparagraph{Case 2.4) $t_r$ is an internal vertex of $P$.}
        This case is analogous to Case 1.4).
    \begin{figure}
        \centering
        \includegraphics[scale=0.52]{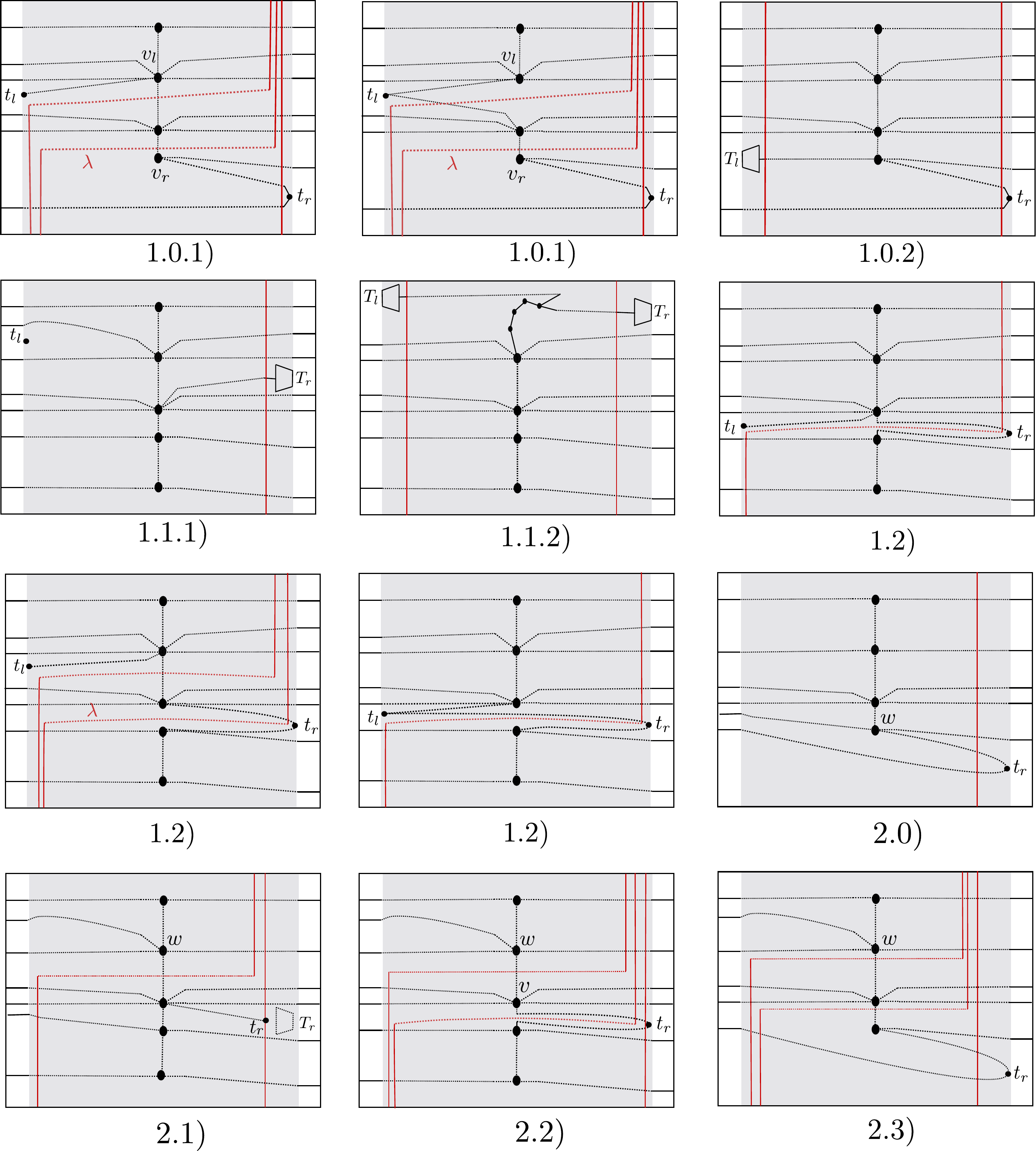}
        \caption{Examples of local disks appearing in cases of Proposition~\ref{l:bpdrawing}, a dashed line shows that we do not know of the exact drawing of an edge, however we know the points that an edge or a segment connects. The dots labeled $t_l$ and $t_r$ show the location of the leftmost and rightmost moves, respectively. The parts of the tree on which they lie are sometimes not drawn for reducing clutter. We also have depicted the bubbled sub-trees.}
        \label{fig:spineexamples}
    \end{figure}
    

\end{proof}

\subsection{Life of a spine disk}
A spine disk is ``born'' either at the very beginning of the decomposition when the spine path is a single vertex or is born, where it is cut out of another spine disk. The second possibility happens in the cases 2.0), 2.1), 2.2), 2.3),  1.0.1), 1.0.2), 1.2) or 1.3). These correspond to the rectangles in the diagrams of these cases. In each case, all but possibly one of the born spine disks has a single vertex.

After a spine disk is born, it grows in two ways. First, the spine path might become longer using the case 1.1.2). Second, new boundary edges might appear attached to the same spine path. This happens by case 1.1.1) followed by 2.1) where there is a single spine disk in the skew spine. Other cases will cut the spine path, or change the left and right boundary edges of the spine disk.

Now if we consider all spines with a fixed path $P$ and fixed four extreme boundary edges, after one such spine is generated, the only possible operations on it are adding boundary edges, that is case 1.1.1) followed by 2.1). Note that each pair of these moves adds a pair of left and right boundary edges incident to the same vertex, with bubbles at their ends. Adding pairs of boundary edges to the spine finishes either when the next move of the spine disk is one of 1.0.1), 1.0.2), 1.1.2) and 1.3) or the spine disks turns into a skew disk and the path is cut in case 2.1) or otherwise the skew disk undergoes other cases than 2.1). In brief, when a spine disk with the spine path $P$ is created, it then continues its life by adding pairs of opposite boundary edges, until either $P$ is cut or is extended to a longer path. If $P$ is maximal, adding pairs of opposite boundary edges and bubbles stops when $P$ is cut into smaller spine paths.

\subsection{Structure tree}
Recursive applications of Proposition~\ref{l:bpdrawing} to an optimal-height drawing of a spine or a skew spine disk, for instance one containing all of $T$, result in an optimal drawing that has a hierarchical structure. 
Any node in the hierarchy is a spine or skew spine disk, and a node is decomposed into its children using one of the possibilities of Proposition~\ref{l:bpdrawing}.
The leaves of the hierarchy are vertex disks.
We call this hierarchy a \emph{structure tree} of the optimal drawing. We call a drawing which has such a structure tree a \textit{structured drawing}. For instance the drawings output by Proposition~\ref{l:bpdrawing} are structured drawings.

\section{Optimizing the optimal-height drawings}

In this section, we first define the quality of drawings. We then consider those optimal drawings that maximize this quality measure.

\subsection{Quality of a drawing}
Let $\phi$ be any drawing of a local disk $D$. We denote by $\Lambda' = \Lambda'(\phi) = \{\lambda_i\}$ the set of combinatorially distinct vertical lines in the plane (that lie in general position with the drawing). That is, the strip $S_i$ bounded by $\lambda_i$ and $\lambda_{i+1}$, after removing pipes, is either: i) a vertex move, 
that is, contains a single vertex and no bends, or ii) contains a single bend and no vertices. Such a set of vertical lines can be chosen for any drawing $h$ in general position.
Consider a strip $S_{ij}$ bounded by $\lambda_i$ and $\lambda_j$. If $S_{ij}$ contains only pipes we remove $\lambda_i$ or $\lambda_j$ (whichever is to the right of the other) and all the lines in between form $\Lambda'$, and repeat this operation. Let $\Lambda= \Lambda(\phi)$ be the remaining set of vertical lines. We again consider strips $S_{ij}$ bounded by $\lambda_i$ and $\lambda_j$ in $\Lambda$.  
If after removing pipes from the strip $S_{ij}$ it becomes a bubble, spine disk, skew spine disk, or bend, we respectively say that $S_{ij}$ is a bubble, spine disk, skew spine disk or bend. Recall that a bubble is a special type of skew spine disk with only one boundary edge in one side.

Note that bubbles are either disjoint or nested and therefore give rise to a hierarchical structure.
We say that a vertex $v$ is of \emph{depth} $d$ if it is contained in exactly $d$ bubbles. That is, there are exactly $d$ strips, bounded by the lines of $\Lambda$, that contain the vertex and that are bubbles.
We define the \emph{depth} of a bubble and a (skew) spine disk analogously to depth of a vertex to be the number of bubbles that properly contain them.

We say that a line $\lambda\in \Lambda(\phi)$ is of depth $i$ if it is contained in exactly $i$ strips which are bubbles, and in none of them it is a boundary. For instance, lines of depth 0 do not cut any bubble. Let $\Lambda_i = \Lambda_i(\phi)$ denote the set of lines of depth $i$. 
Let $\Lambda_{i,j}=\Lambda_{i,j}(\phi) \subset \Lambda_{i}$ be the set of lines of height $j$ and depth $i$, and let $\delta_{i,j}(\phi)= |\Lambda_{i,j}|$ be the number of lines of depth $i$ and height $j$. Note that $\delta_{i,j}(\phi)$ can be 0. Let $\Delta_i(\phi)$ be the sequence $(\delta_{i,0}(\phi), \delta_{i,1}(\phi), \ldots)$,
and define the \emph{quality} of the drawing $\phi$ as
\[Q(\phi) = ( \Delta_0(\phi), \Delta_1(\phi), \ldots).\]

For two drawings $\phi$ and $\phi'$, we compare their complexities $Q(\phi)$ and $Q(\phi')$ lexicographically, where we also compare the sequences $\Delta_i(\phi)$ to $\Delta_i(\phi')$ lexicographically.
Specifically, a drawing of maximum quality maximizes the depth sequences $\Delta_i$ from left to right.
That is, we are interested in the drawings where the sequence $\Delta_0$ is maximized, and among these the sequences where $\Delta_1$ is maximized, and so on.
We emphasize that maximizing the quality does not necessarily minimize the height of the drawing. Instead, we merely use the quality measure to reduce the search space for minimum height drawings..

There remains still some arbitrariness in optimal drawings with maximum quality. For instance, a star with $2k$ leaves and a central vertex can be drawn with optimal height and with maximum quality in an exponentially many different ways, giving rise to exponentially many spine disks, by changing the order of the vertices. We get rid of these choices using the notion of secondary quality to be defined in Section~\ref{ssec:purturbed}.

By Lemma~\ref{l:simpli}, there exists an optimal simplified and balanced drawing of any local disk $D$, and by Lemma~\ref{l:finiteness}, $(H+1)n$ is an upper bound on the complexities of simplified drawings with height $H$. Therefore, the set of quality sequences of the set of all optimal, balanced and simplified drawings of a local disk is non-empty.
Moreover, each quality sequence for such a drawing consists of at most $H(H+1)n$ terms $\delta_{i,j}$, since the depth is at most the number of lines in $\Lambda$ and each depth-sequence $\Delta_i$ contains at most $H=O(n)$ different height values.

\subsection{Properties of drawings with maximum quality}

The following lemma implies that by restricting the dynamic programming to disks that can appear in a minimum quality drawing we do not lose any (skew) spine disks necessary for computing the optimal height. We need this since the dynamic program only constructs structured drawings.


\begin{lemma}\label{l:comdecom}
Let $\phi$ be a simplified, balanced drawing of a spine disk $D$ that has maximum quality $Q(\phi)$ over all drawings with the same height as $\phi$. Then $\phi$ is a structured drawing.
\end{lemma}

\begin{proof}
We apply the spine decomposition algorithm to the input drawing $\phi$ and show that if the algorithm changes the drawing, then the quality decreases, while we maintain the balanced and simplified properties of the drawing. Observe that the decomposition algorithm changes the drawing in two ways. Either by bubbling, or by making a shortcut of a sub-drawing vertical using Lemma~\ref{l:pausing}. 

Assume at some stage the algorithm bubbles a sub-tree from a disk $D$. Without loss of generality we assume the bubble is on the right side of the disk. Assume the disk has $b$ right boundary edges and is at depth $d_0$, and there are $p$ pipes above and below this disk in the drawing. First, note that no vertical line of depth $d_0$ increases its height by the bubbling procedure. 
Second, as in Figure~\ref{fig:bubbling}, a new line of height $p+b+1$ is introduced to the vertical lines of depth $d_0$ (If the bubble does not exist already. This line is the line at the center of the right part of Figure~\ref{fig:bubbling}). There are two lines of height $p+b$ in the disk $D$, namely the boundaries and the number of these lines does not change. 
If we show that no original lines of height $p+b+1$ decreases its height, we could deduce that the bubbling has improved the quality, since the new drawing now has one more lines of height $p+b+1$ in depth $d_0$, while the numbers of lines of lower depths and those with depth $d_0$ and lower heights are not changed. But this is trivial since after bubbling there is still exactly two lines of height $p+b$ namely the boundaries of the new drawing, and no line in the interior of the right hand side of Figure~\ref{fig:bubbling} has decreased its height to $p+b$. In other words, every vertical line that originally crossed any edge other than boundary edges of $D$, still crosses some such edge edge. Note that these might have become equal to each other. In brief, bubbling increases the quality.


Since the drawing inside the bubble is a (copy of a) subset of the original drawing, it is balanced.
It is easily checked that the union of the bubble and the rest of the drawing is also balanced. It is also easily checked that other than bend-bend cancellation and stuck-slide, other simplifying moves replace lines with lines of lower height increasing the quality. The lines in a strip that contains a bend-bend cancellation do not count toward quality. The stuck-slide keeps the number of lines with any height fixed. However, after such a move always a vertex-bend separation or another stuck-slide is possible, namely involving the vertex or bend which is the cause of the stuck. If we encounter a bend which causes a stuck but it itself is non-stuck it has to be incident to a vertex (as in the top right of Figure~\ref{fig:simplification}). This vertex defines a vertex-bend separation with the original bend. This will guarantee that eventually the quality increases. Therefore, applying non-strong simplification moves to make the new drawing simplified only improves the quality, while keeping the drawing balanced. Since the resulting drawing would be simplified, balanced and with better quality, we reach a contradiction. Therefore the algorithm does not change the drawing by bubbling.


Next consider pausing at a shortcut in some strip $S$ and assume that the depth of the disk which is cut with shortcuts is $d_0$.
In all cases of Figures~\ref{fig:spinedecom1} and~\ref{fig:spinedecom2}, the vertical lines separating the blocks in the figures are maximal disjoint sets of shortcuts uniquely defined by knowing the edges and vertices of $t_l$ and $t_r$. These cuts have some height $b'$. When applying the decomposition algorithm, we cut the disk at all cuts of length $b'$ that are already vertical. If no cut remains, then this step of the algorithm does not change the drawing. Otherwise, there is a remaining shortcut that is disjoint from all vertical cuts of length $b'$. Pausing at the shortcut then increases the number of lines of height $b'$ and depth $d_0$ (Lemma~\ref{l:pausing}) without changing those of smaller height (since only the boundary of the strip of $D$ is shorter than $b'$) or smaller depth. 
After pausing at a shortcut, we might have made the drawing non-balanced and non-simplified. 
We apply balancing and non-strong simplifying moves. We first apply balancing moves, if necessary. The balancing move, Figure~\ref{fig:balance}, does not decrease the quality, it only adds one line which increases the quality. The bend-bend cancellation makes two lines equal. However, the two of them cannot be our shortcuts since between our shortcuts there are vertices. It follows that cancelling bends does not make the number of lines of height $b'$ smaller (in other words, the line of height $b'$ exists in $\Lambda$). Also, the non-strong simplifying moves only improve the quality. Therefore, the result would be simplified and balanced everywhere and has a better quality than $h$ which is again a contradiction.

It follows that the algorithm does not change the drawing.

\end{proof}


Recall that for a path $P$ the set of anchor edges, $A(P)$, is the set of edges which have exactly one endpoint on the path $P$. Also, the set of anchor edges of a spine disk $D$, $A(D)$, is the set of anchor edges of the spine path of $D$.
Let $D$ be a (balanced) spine disk with $2b$ boundary edges and $e \in A(D)$ be an anchor edge of $D$. Let $H(D)$ denote, as always, the optimal height of the disk $D$ and $eH(T',e')$ denote the optimal exposed height of a sub-tree $T'$ with respect to the edge $e'$. Also recall that the sub-tree $T_e$ anchored by $e$ is the sub-tree rooted at the endpoint of $e$ which is not in the spine of the disk $D$. We say $e$ is \emph{light} (with respect to $D$) if the exposed height of the sub-tree $T_e$ satisfies $eH(T_e,e) \leq H(D)-b+1$. 

The significance of light edges is that if we know a boundary edge $e$ of $D$ is light then given any drawing of $D$ we can redraw the sub-tree $T_e$ near the boundary of $D$ in a small bubble without increasing the height, since the maximum height over the bubble would be $eH(T_e,e)+b-1$ which would be at most $H(D)$.

\begin{lemma}\label{l:lightedges}
Let $D$ be a spine disk of depth $d$ in a drawing $\phi$ with maximum quality $Q(\phi)$. Let $e\in A(D)$ be a light edge and $T_e$ the sub-tree anchored by $e$. If $e$ is a boundary edge then $T_e$ is drawn in a bubble of depth $d$ with $e$ as the single boundary edge. Moreover, the strip between the bubble of $T_e$ and the disk $D$ is a sequence of bubbles of depth $d$ anchored at the spine of $D$, or bends.
\end{lemma}

\begin{proof}

Without loss of generality we assume that $e$ is a right boundary edge. Assume on the contrary that $T_e$ is not drawn in a bubble of depth $d$. Moreover, assume $D$ is the leftmost spine disk in the the drawing such that the lemma is not true for it.

Let $r$ be the root of $T_e$. Assume $r$ is in $k$ bubbles. Suppose that among these $k$ bubbles there are two bubbles $B_1$ and $B_2$ such that $B_1$ contains $T_e$ and $B_2$ does not contain $T_e$, as in Figure~\ref{fig:twobubbles}. By construction, the bubbles are either nested or have no vertex in their intersection. Since $B_2$ contains $r$ and not $T_e$ it contains the sub-tree other than $T_e$ anchored by $e$. If $T_e$ is in $B_2$ then $B_2$ must contain all of the tree $T$ which cannot be. Thus it has to be that $T_e$ is in $B_1 \cap B_2$. It follows that if $r$ is in $k$ bubbles and one of them contains $T_e$ then $T_e$ is in $k$ bubbles.

Note that a bubble that contains $r$ and does not contain exactly $T_e$ has to contain $T'_e$, the second sub-tree anchored at $e$. Therefore, if $r$ is in $d+1$ bubbles $d\geq1$, and none of them contains exactly $T_e$, then $D$ is in $d+1$ bubbles which cannot be. It follows that at least one of them must contain exactly $T_e$. It follows that, $T_e$ then is drawn in a bubble that is contained in $d$ bubbles and the first statement of the lemma is proved in this case, if $r$ is in $d+1$ bubbles. If $d=0$ and $r$ is in $d+1$ bubble then that bubble must contain $T_e$. It follows that in this case too the first statement is proved. In the rest of the argument we assume $r$ is not in $d+1$ bubbles with $d\geq 0$ and reach a contradiction.

\begin{figure}
    \centering
    \includegraphics{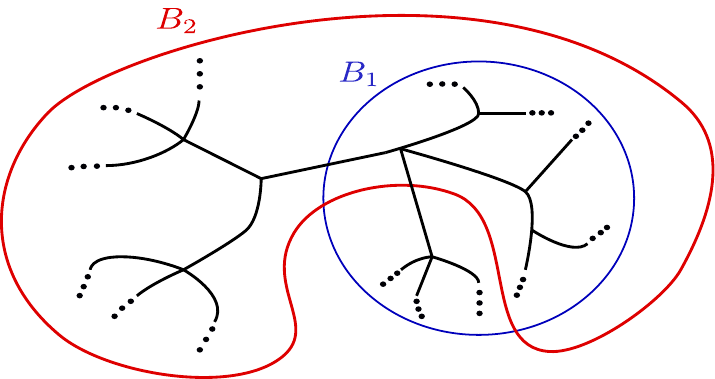}
    \caption{Intersecting bubbles.}
    \label{fig:twobubbles}
\end{figure}
\begin{figure}[b]
    \centering
    \includegraphics{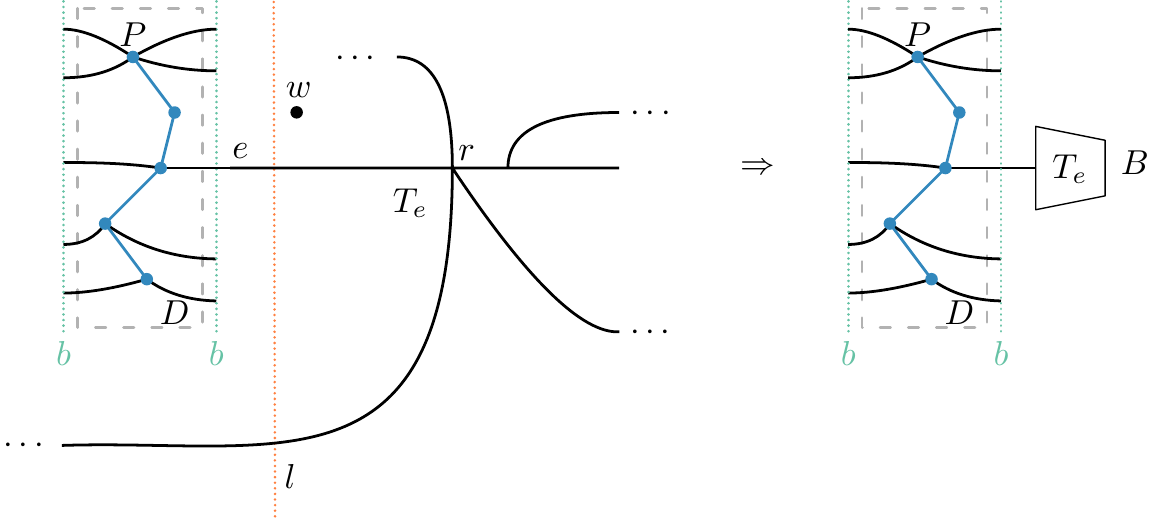}
    \caption{Bubbling sub-trees anchored via light edges.}
    \label{fig:bubbledsubtree1}
\end{figure}

Since edge $e$ is light the point where the edge $e$ intersects the boundary of $D$ is an admissible spot for the bubble of $T_e$ in the sense that we can change the drawing by drawing $T_e$ in a bubble near this point, without increasing the optimal height of the drawing, see Figure~\ref{fig:bubbledsubtree1}. 

Let $d'$ be the minimum depth of the vertices of $T_e$, thus $d'\leq d$ since $r$ is of depth at most $d$. In the new drawing $h'$, all these vertices are in at least $d+1$ bubbles, namely $d$ bubbles containing $D$ and one bubble containing $T_e$. Thus all the vertical lines that cross the bubble of $T_e$ are of depth at least $d+1$. The vertical lines of $h$ outside of the bubble of $T_e$ have decreased their heights or remained constant. It might be that two distinct lines become equal.
Certainly some of the vertical lines of depth at most $d$ has now less height, for instance the lines near the vertex of $T_e$ of depth $d'$. This implies that the new drawing has improved quality. 
It is easily checked that if the drawing of $T_e$ inside the bubble is chosen to be balanced and simplified then $h'$ is balanced and simplified. But this contradicts our assumption that $h$ has a fat structure.

To prove the second statement of the lemma, assume that there is a vertex $w$ between the disks $B$ and $D$ which is not in a bubble of depth $d$. This implies there is a line $l$ next to the vertex which is not in a bubble of depth $d$, that is whose depth is at most $d$.

Now again we change the drawing as above by drawing $T_e$ near the boundary. Consider the effect of this change in the drawing on the heights of lines. The height of the line $l$ in the new drawing has decreased by at least 1. Any other height of a line is decreased or is constant other than those lines which are inside the new bubble of $T_e$. However, these are of depth $d+1$. Therefore the new drawing has better quality which gives a contradiction.

Therefore, every vertex between $B$ and $D$ is in a bubble of depth $d$. This bubble does not contain $D$ nor $B$, thus lies in the strip between $B$ and $D$. It is easily seen by inspecting the decomposition lemma that these bubbles cannot be anchored in the opposite direction than $B$, that is, to a bubble on the right of $B$. On the other hand, if there is a bubble $B'$ that is not anchored at $D$, but at a disk to the left of it, then this contradicts our choice of $D$ and $e$. Hence all the bubbles are anchored at the disk $D$.

\end{proof}

\subsection{Breaking ties while respecting the orders}\label{ssec:purturbed}
In our arguments we will use a tie-breaking mechanism to decide between optimal drawings which all have maximum quality. We first define a perturbation of the original heights. 

\begin{figure}[h]
    \centering
    \includegraphics{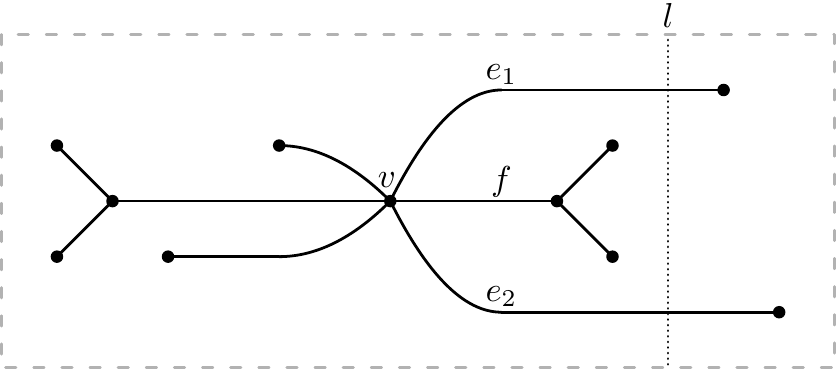}
    \caption{The edge $f$ is sandwiched between $e_1$ and $e_2$ with respect to $l$.}
    \label{fig:sandwich}
\end{figure}
Let $v$ be a vertex in $D$ and let $e_1,e_2$ and $f$ be edges incident to $v$. Let $l$ be a vertical line. We say that $f$ is \textit{sandwiched} between $e_1$ and $e_2$ with respect to $l$ if $f$ does not intersect $l$ but $e_1$ and $e_2$ intersect $l$, and $f$ lies in the resulting bigon, see Figure~\ref{fig:sandwich}. We add a small $\epsilon$ ($0<\epsilon\ll 1$) for every sandwiched edge with respect to $l$ to the height of $l$. The resulting value is called the \textit{perturbed height} of $l$.

Consider the set $\Lambda$ of lines of a given drawing as defined above and let $W(\phi)=(w_1, w_2, \ldots )$ be the sequence of perturbed heights of lines in $\Lambda$, sorted in a non-decreasing order. Let $\phi_1$ and $\phi_2$ be two drawings with maximum quality. We say that $\phi_1$ has a better \textit{secondary quality} than $\phi_2$ if the sequence $W(\phi_1)$ is lexicographically smaller than $W(\phi_2)$.

\subsection{Fat structures}

Let $D$ be a local disk. Let $\phi$ be an optimal, simplified and balanced drawing such that $Q(\phi)$ is maximal among such drawings and also its secondary quality is the best possible. By Lemma~\ref{l:comdecom}, such a drawing has a structure tree. We call the resulting structure a \emph{fat structure}\footnote{The name comes from the fact that the bubbles in a minimal drawing tend to contain a maximal part of the tree.} for $D$. It follows that any drawing which is optimal, simplified and balanced has to have worse or equal quality (or equal quality and equal or worse secondary quality).


The proof of the following is straightforward.
\begin{lemma}
Let $\phi$ be a drawing with a fat structure of a local disk $D$. Then for every local disk $D'$, corresponding to a node $N'$ in the structure tree, the restriction of the structure to $D'$ is a fat structure.
\end{lemma}

The following lemma allows us to enumerate the spine disks which are possible in a fat structure. It is therefore the main result underlying our algorithm. Refer to Figure~\ref{fig:prop10} for an example.

\begin{figure}
    \centering
    \includegraphics{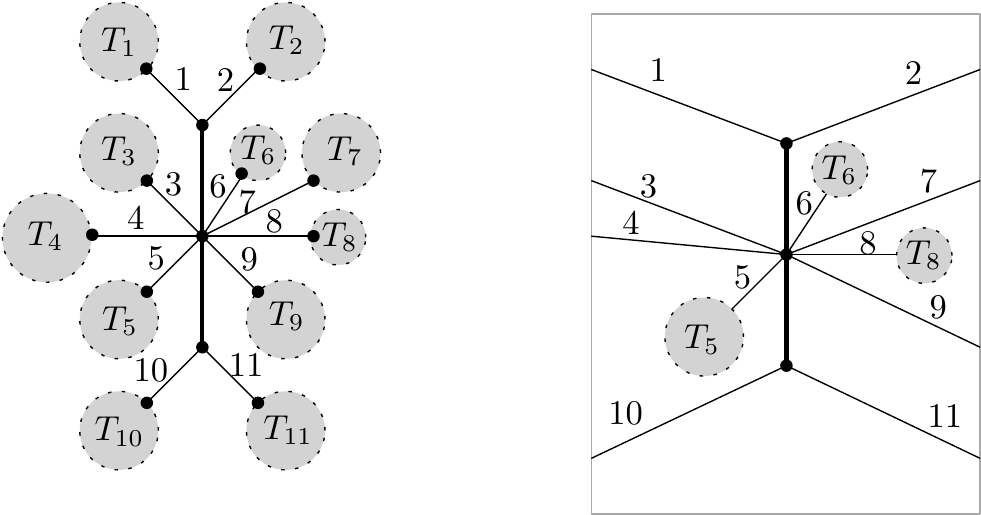}
    \caption{Left: the path $P$ (thick edges), anchor edges $A(P)$ numbered 1 to 11, and the anchored sub-trees. Right: A spine disk with spine path $P$ and $b=3$. The proposition statements imply that 1) Edges 5, 6, 8 are light. 2) Among the six sets $\{1\}$, $\{2\}$, $\{3,4\}$, $\{7,9\}$, $\{10\}$, $\{11\}$ at most one can contain a light edge. Assume it is $\{7,9\}$. 3) If $eH(T_7)=eH(T_9)=H(D)-3$ then $eH(T_8)\neq H(D)-3$.
    If $eH(T_7)<H-3$ and $eH(T_9)<H-3$, then $eH(T_8)\geq H-3$ by the second statement of part 3.
    }
    \label{fig:prop10}
\end{figure}

\begin{proposition}[Characterization of Spine Disks in Fat Structures]\label{p:spinechar}
Let $P$ be a path in the tree and let $D$ be a spine disk with spine path $P$, such that $D$ is a node in a fat structure.
Let $b>0$ be the number of left (equivalently right) boundary edges of $D$.
\begin{enumerate}
    \item Every edge $e \in A(P)$ that lies entirely in the interior of $D$ is light.
    \item All light boundary edges of $D$ are incident to a single vertex $v$, and intersect the same (left or right) boundary.
    \item For $\eta\geq 0$, let $E({\eta}) \subset A(P)$ be the set of anchor edges of $P$, incident to $v$, for which the exposed height of the sub-tree anchored by that edge is $H(D)-b+1-\eta$.
    Then, for $\eta=0$, if any edge $e$ in $E(0)$ is not a boundary edge, then $e$ is not sandwiched, with respect to the boundary lines, between two edges of $E(0)$ that are boundary edges.
    Moreover, if any edge $e$ of $E(\geq 1):=\bigcup_{j\geq 1} E(j)$ is not a boundary edge, then $e$ is not sandwiched between two edges of $E(\geq 1)$ that are boundary edges.
\end{enumerate}
\end{proposition}

We remark that the secondary quality is only needed in the proof of the second part of statement 3. That is, the rest of proposition is true for drawings with maximum quality.

\begin{proof}
Let $f\in A(P)$ be an anchor edge of $P$ in the interior of $D$.
We claim that the exposed height of $T_f$ is at most $H(D)-b+1$. To see this, note that $T_f$ can be exposed (that is, drawn in a bubble with $f$ as boundary edge) by taking as the mage of $T_f$ the current drawing of $T_f$ inside $D$ and as the image of $f$ the current image of $f$ in $D$ concatenated with image of a path from $v$ to the boundary of $D$ that uses a single boundary edge. This exposed drawing of $T_f$ is within $D$ and does not use $b-1$ edge disjoint paths from the left to the right boundary in $D$. It follows that its height is at most $H(D)-(b-1)$. Therefore $f$ is a light edge. 

We now prove the second statement. Assume there are light edges $e_l$ and $e_r$ among the left and the right boundary edges of $D$ and that $e_l$ is incident to $v$ and $e_r$ is incident to $w$ and $v\neq w$. Lemma~\ref{l:lightedges} implies the sub-trees $T_{e_l}$ and $T_{e_r}$ are bubbled, and we can move the bubbles of $T_{e_l}$ and $T_{e_r}$ near the left and right boundary of $D$ without changing the non-perturbed heights. But then there exist a set of shortcuts as in Figure~\ref{fig:spinev1} in a larger spine disk $D'$. We can therefore pause the drawing of $D'$ at these shortcuts. Let $d'_0$ be the depth of $D'$. The heights of lines of depth less than $d'$ are not changed. And we have not created any line of height $b-2$ or less since the minimum height over the disk $D'$ is always $b-1$. It follows that $\delta_{d'_0, b-1}$ increases while the lower terms remain unchanged. As in the proof of Lemma~\ref{l:comdecom} we can make the new drawing balanced and simplified without making the number of lines of $\delta_{d'_0, b-1}$ smaller. The resulting drawing has a better quality than a fat structure which is a contradiction.

\begin{figure}
    \centering
    \includegraphics{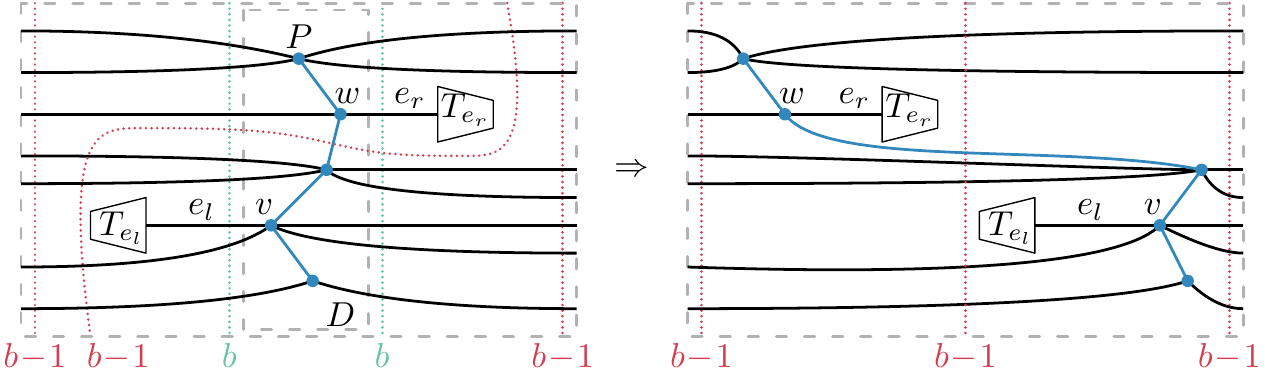}
    \caption{Light left and light right boundary edges must be incident to the same vertex.}
    \label{fig:spinev1}
\end{figure}

Next, assume there are light edges $e_1$ and $e_2$ among right boundary edges such that $e_1$ is incident to $v$ and $e_2$ is incident to $w$ and $v\neq w$. Without loss of generality assume $e_2$ anchors the sub-tree which is bubbled after the sub-tree anchored by $e_1$. This implies that $e_1$ is paired with a left boundary edge $f$ with a bubble. Since otherwise the spine would have been cut (since it has at least two vertices). For the same reason, after bubbling $e_2$, in the resulting skew spine disk, only the case 2.0 and 2.1 has happened with no cuts. Case 2.0 is also impossible since it would create a skew spine disk. It follows that there is also a bubble opposite to $e_2$. Now, after a possible change in the order of the bubbles of $e_1$ and $e_2$, we can make sure that there is a shortcut, see Figure~\ref{fig:spinev2}. Similar to the last paragraph, this will reach a contradiction.

\begin{figure}
    \centering
    \includegraphics{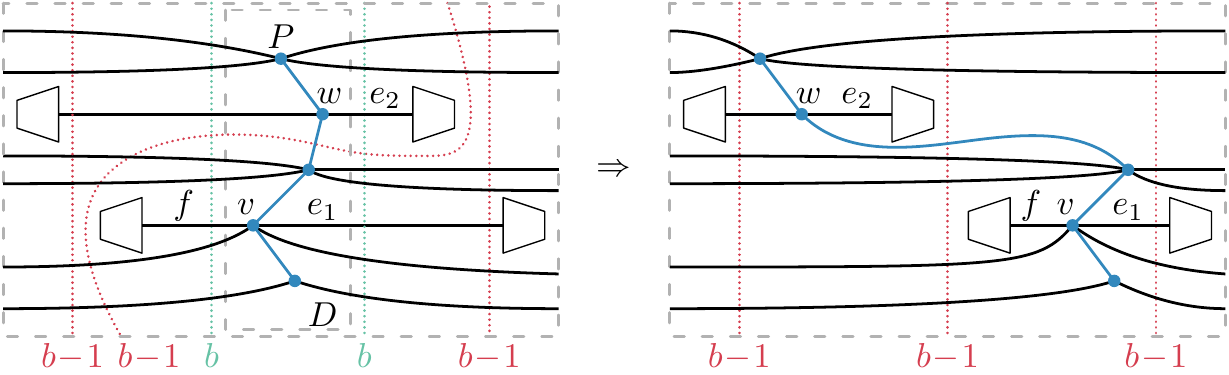}
    \caption{Light right boundary edges must be incident to the same vertex.}
    \label{fig:spinev2}
\end{figure}

We now prove the third statement. Assume that in some $D$ there are two right light boundary edges, $e_1$ and $e_2$ incident to $v$, that are consecutive on the right boundary of $D$ but not around $v$. Let $t$ be the rightmost point (which is not on $e_1$ and $e_2$) of the drawing in the region $R$ bounded by $e_1$, $e_2$ and the right boundary of $D$. Let $T_f$ be the sub-tree anchored at $v$ that contains $t$ and let $f$ be the edge anchoring $T_f$ to $v$. Since $f$ is in the interior of the tree, $eH(T_f) \leq H(D)-b+1$ by the first statement. 
We claim that $eH(T_f) \leq H(D)-b$. Note that $R$ contains a single vertex, namely $v$, on its boundary, see Figure~\ref{fig:subtreeexposure2}. Let $p$ be the leftmost point of the disk $R$. If $p=v$ or a point on the spine then the claim is clear, since we have not used any boundary edge for exposing $f$. Otherwise, there is a path on the boundary of $R$ from $v$ to $p$. We can expose the sub-tree $T_f$ using this path. Now we claim that in any vertical line of the drawing, there is at least $b$ points of the drawing of $D$ which are not in the exposed drawing of $T_f$. Indeed, we have not used $b-1$ pairs of opposed boundary edges. Moreover, we have only used a portion of an edge $e=e_1$ or $e_2$, from the point $v$ to an extreme point of that edge on the left. The union of the edge opposite $e$ and the path from $p$ to the point where $e$ intersects the right boundary covers all of the $x$-interval of the disk. Because we have not used any point of these two paths this proves the claim. It follows that $e_1,e_2 \in E(0)$ are consecutive along the boundary, and therefore none of them is sandwiched between two other ones.

Let $e_1\in E(k_1), e_2 \in E(k_2)$, $k_1,k_2 \geq 1$, be right boundary edges that are consecutive along the boundary but not around $v$. Moreover, assume that this pair is the first such pair in the order from top to bottom along the right boundary and that $e_1$ is above $e_2$. Let $f$ and $T_f$ be as above a counterexample, $eH(T_f)\leq H(D)-b$ and $t$ as above be the rightmost point in the region $R$ defined as before. In this case, we can draw $T_f$ in a bubble near the boundary of $D$ and inside $D$ without increasing the maximum height. If $T_f$ is not already bubbled in the drawing, this operation increases the quality (as in proof of Lemma~\ref{l:comdecom}) which cannot be. Therefore, $T_f$ is drawn inside $D$ in an exposed way. It follows that we can exchange the order of bubbles of $T_f$ and $T_1$. No line of depth $d-1$ or less changes its height after this change, while the perturbed height of some line of depth $d$ (e.g. the one just before the bubble of $T_f$ in the new drawing) is decreased and the perturbed height of other lines does not increase, see Figure~\ref{fig:spinev3}. This is a contradiction with our choice of the drawing. This finishes the proof of the third statement.

\begin{figure}
    \centering
    \includegraphics{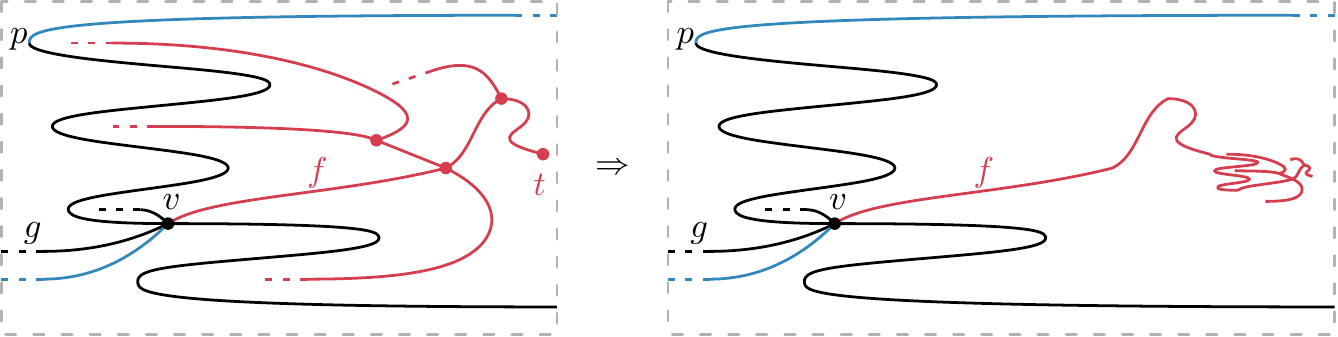}
    \caption{Exposing a sub-tree sandwiched between boundary edges.}
    \label{fig:subtreeexposure2}
\end{figure}

\begin{figure}
    \centering
    \includegraphics{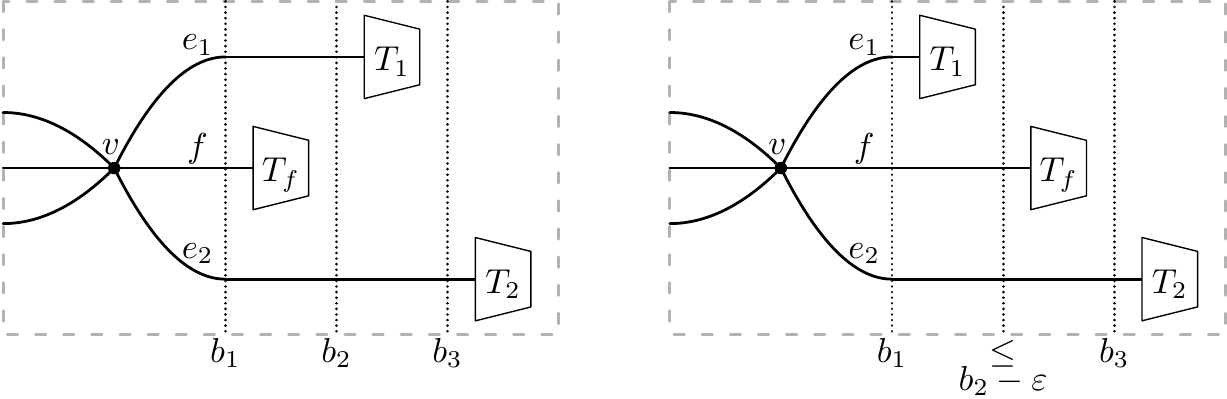}
    \caption{Changing the order of bubbles.}
    \label{fig:spinev3}
\end{figure}

\end{proof}

\section{The dynamic program}

We describe the algorithm for computing the optimal height of an input drawing. Modifying the dynamic program to compute an actual optimal height drawing is standard. 

We think of row $m$ of the dynamic programming table as containing (the description) of those spine and skew spine disks that have exactly $m$ vertices in their interior and satisfy Proposition~\ref{p:spinechar}, together with their optimal heights. For $m=1$, i.e. the first row, we must consider the (skew) spine disks with exactly one vertex in their interior. Since we are interested in balanced drawings, we know that each vertex $v$ of even degree defines $O(d(v))$ distinct spine disks, where $d(v)$ is the degree of the vertex $v$. These are given by all the $O(d(v))$ possible balanced partitions of the edges incident to $v$ into left and right edges, maintaining the order around $v$. The optimal drawings are trivial. Similarly, vertices with odd degree determine $O(d(v))$ distinct skew spine disks. 

Assume that we have populated the table up to row $m-1$. The algorithm first computes all spine and skew spine disks with $m$ vertices that satisfy Proposition~\ref{p:spinechar}. 
Then for each computed (skew) spine disk the algorithm computes its optimal height.

\subparagraph{Computing all (skew) spine disks with $m$ vertices}
We consider only spine disks and this also determines all possible skew spine disks. This is because a skew spine disk is the result of changing one non-boundary anchor edge of a spine disk into a boundary edge.

If we know the exposed heights of anchored sub-trees, then Proposition~\ref{p:spinechar} implies that a spine disk is determined uniquely given the following parameters. We also indicate an upper bound on the number of possibilities for each of them.
\begin{enumerate}
    \item The spine path $P$: $O(n^2)$ possibilities.
    \item The boundary height $b$: $O(n)$ possibilities.
    \item The height $H$: $O(n)$ possibilities.
    \item A partition of $A(P)$ into cyclically contiguous subsequences $A_L(P)$ and $A_R(P)$: $O(n^2)$ possibilities.
    \item The vertex $v$ to which light boundary edges are incident: $O(n)$ possibilities.
    \item Two consecutive sequences of edges around $v$: one for $E(0)$ and the other for $E(\geq 1)$: $O(n^4)$ possibilities.
\end{enumerate}

For the sake of simplicity, we are very generous in our analysis. As above, there are polynomially many possible values for all these parameters, namely $O(n^{11})$. Not any such set of values defines a valid spine disk and there are some exceptional situations.
The details of how we test each set of parameters is as follows.

We use different parts of Proposition~\ref{p:spinechar} in order. We first choose the spine path $P$, containing $p\geq 1$ vertices, and the number of boundary edges $b>0$. 
All spine disks with $b=0$ have a single (arbitrary) vertex as $P$, and the disk contains all of $T$, so $T$ has $m$ vertices.
Next, consider $b>0$. We must compute all possible sets of boundary edges for disks with $m$ vertices. 

Let $A=A(P)$ be the set of anchor edges of $P$. There is a uniquely defined clockwise order on $A$, obtained by contracting $P$ to a point and considering the clockwise order of the edges around the contraction point. We choose one partition of edges of $A=A_L \sqcup A_R$ into left and right edges, respecting the order. There are $O(n^2)$ such partitions. Up to this point we know the left and right potential edges of a spine disk, however, we do not which anchor edge incident to $P$ is a boundary edge of a spine disk.

Let $f \in A(P)$ and $T_f$ be the sub-tree anchored by $f$. Recall that anchored sub-trees do not contain any vertices of the spine path, so if $T_f$ contains more than $m-p$ vertices, then $f$ has to be a boundary edge. Let $B_1$ be the set of boundary edges determined in this way. If the number of boundary edges found is $2b$ such that there are $b$ edges on each side we have found a spine disk. If it is greater than $2b$ or is equal to $2b$ but one side has more than $b$ edges there cannot be a spine disk with these parameters. In both cases we stop and continue with the next set of parameters. Note that in the former case, there is a single spine disk with the given parameters, that is, the path $P$, the number $b$ and the given partition of $A(P)$ into the left and right edges. 

At this point in the algorithm we apply Proposition~\ref{p:spinechar} to compute all possible spine disks in a fat structure with our parameters. We do not know the height of $D$ that appears in Proposition~\ref{p:spinechar}. By Lemma~\ref{l:finiteness} this height is upper bounded by a constant times $n(H+1)$, hence we can consider all possible values of $H$.
We apply the following for increasing values of $H$, starting from $H=b$, making $H$ also a parameter. 

Let $f \in A(P)-B_1$ be an edge such that $T_f$ has at most $m-p$ vertices. Then the bubble of $T_f$ is a skew spine disk in row $m-p$ of our table. We read off the exposed height of $T_f$ from the table. If $f$ is not a light edge (with respect to $H$) we make it a boundary edge (left or right depending on whether $f$ is in $A_L$ or $A_R$). These boundary edges are correct by part 1 of Proposition~\ref{p:spinechar}. Let $B_2$ be the set of boundary edges computed up to this point. The anchor edges that are undecided at this point are therefore all light edges.
 
If $|B_2\cap A_L| > b $ or $|B_2\cap A_R| > b$, then we stop and deduce that there is no spine disk with our current set of parameters. If both $|B_2\cap A_L| < b $ and $|B_2\cap A_R| < b$, then we again stop since we know by Proposition~\ref{p:spinechar} that light edges appear only on one side of the boundary. It follows that if we have not stopped, at least one of $B_1$ or $B_2$ consists of exactly $b$ edges. If both have exactly $b$ edges the given parameters define a unique spine disk. Otherwise, exactly one of these sets consists of fewer than $b$ boundary edges and the other consists of exactly $b$ boundary edges.

Without loss of generality, assume that $|B_2 \cap A_R|<b$ and set $c_r = b - |B_2\cap A_R|>0$. We need to find yet $c_r$ boundary edges among the light edges of $A_R$. We choose a vertex $v$ of $P$. We then choose a segment of $0 \leq k \leq c_r$ consecutive edges, around $v$, all having exposed height $H-b+1$. Then for the rest of the $c_r-k$ edges, we choose\footnote{This sequence of edges can also be computed by taking the sequence that when interleaved with other boundary edges around $v$ produces the smallest perturbed heights. However, for ease of exposition, we do not apply such optimizations.} a sequence of $c_r-k$ consecutive light edges whose anchored trees have exposed height at most $H-b$ around $v$ (note that the order on $A_R$ is not cyclic anymore). There are $O(n^5)$ total choices for the vertex $v$ and the two sequences of light edges around it. Part 2 and 3 of Proposition~\ref{p:spinechar} imply that these parameters describe a unique spine disk of a fat structure, if at all. If we cannot continue at any point in this procedure, there are no spine disks with the parameters at hand.

\subparagraph{Computing the optimal height of spine disks with $m$ internal vertices.}
In the previous section, we computed possible spine disks with $m$ vertices using the information in the dynamic programming table. Here we compute the optimal height of these disks.

Let $D'$ be a spine disk in the table with $m$ internal vertices. Proposition~\ref{l:bpdrawing} gives a number of possibilities that a structured drawing of $D'$ can be decomposed into drawings of spine and skew spine disks of lower complexity (i.e. number of vertices and edges). As long as we know what the first and the last moves in the drawing of $D'$ are, we know which case of the proposition applies. The first move either is a bend or a vertex. If it is a bend it can be on any edge. Thus the total possibilities for the first (or the leftmost) move is the total number of edges and vertices in the tree $T_D$. Analogously the rightmost move has the same maximum number of 
possibilities. If $D'$ contains all of $T$, it is easily observed that the first and last move can only be vertices. It follows that there are $O(n^2)$ possible choices for first and the last move of the optimal drawing of $D'$. For each such choice, the algorithm computes the sequence of cuts (that is, vertical lines separating the internal blocks in Figures~\ref{fig:spinedecom1} and~\ref{fig:spinedecom2}) for the corresponding case, and then reads off the optimal height for internal disks in the corresponding case of Figure~\ref{fig:spinedecom1} from the table. This is possible since all internal (skew) spine disks in the cases have fewer vertices. It then computes the height for the possibility under consideration. This amounts to writing a number in an internal trapezoid or rectangle in the figure, as the height of that block, and then computing the resulting height of the larger (dashed) spine disk. The best height obtained, for all the choices of first and last moves, is written in the table as the optimal height of the disk $D'$. Since the children of $D'$ in a fat structure are guaranteed to be found among the spine disks in the table, we see that we obtain an optimal height for the disk $D'$. 

The optimal height for the skew spine disks can be computed similarly with the only difficulty being that in case 2.0) of Figure~\ref{fig:spinedecom2} it might happen that the resulting spine disk has the same number of vertices. In this case, there are four possibilities that result from applying $x$- and $y$-reflections to the figure. In each of these possibilities, case 2.0) can be repeated, resulting in a nested set of skew spine disks. If we choose $s$ to be the number of nested such skew spine disks, there is a unique configuration that is possible for each value of $s$. The case analysis for $s=3$ is depicted in Figure~\ref{fig:case2-0}. There can be at most $H$ such iterations for each possibility. Therefore, we have $4H$ possibilities here before going into a skew spine disk where the next move is not case 2.0) and the problem breaks into disks with less vertices. It follows that this special case does not change the fact that the algorithm runs in polynomial time. This finishes the description of the algorithm.

\begin{figure}
    \centering
    \includegraphics{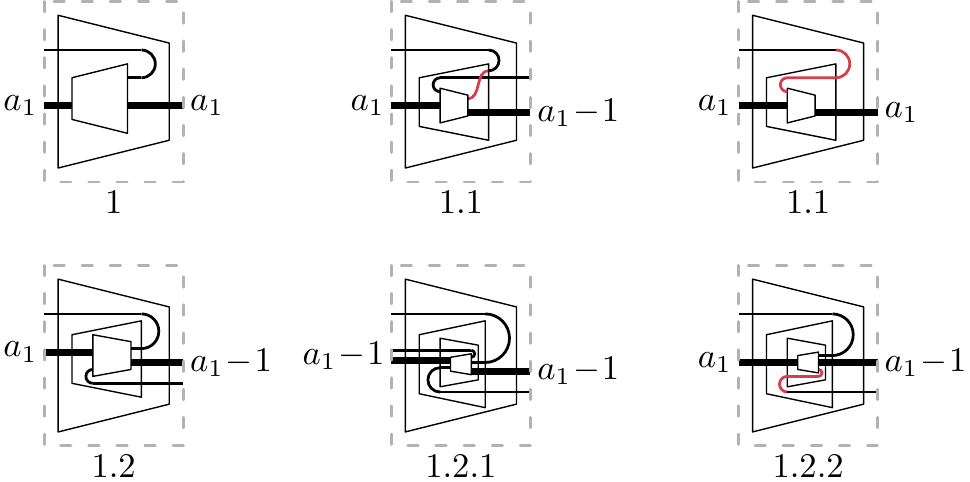}
    \caption{The figure labeled 1 is one of the four possibilities for case 2.0). The figures labeled 1.1 are the result of one nesting where the bend of the innermost case 2.0) is on top. One of the cases 1.1 forces an intersection, the other a bend-vertex separation or a stuck-slide with the part of drawing inside the innermost trapezoid. Hence they are both impossible.  The bend can be on the bottom as in case 1.2. The two possibilities for one more nesting of case 2.0) are labeled 1.2.1 and 1.2.2. Only 1.2.1 is possible. We have not depicted an impossible alternative for 1.2.2 that contains a crossing analogous to 1.1. The pattern continues for deeper nestings. }
    \label{fig:case2-0}
\end{figure}

We have therefore proved the following.

\begin{theorem}
Let $D$ be a (skew) spine disk. There is a polynomial-time algorithm for drawing $D$ with optimal height.
\end{theorem}



\section{Discussion}
    We have presented the first polynomial-time algorithm for drawing plane trees with optimal height.
    The case of weighted plane trees remains open.
    Moreover, the setting of unweighted graphs remains open, but is believed to be NP-hard by some.
    However, we believe that a polynomial time algorithm may exist even in this setting.
    
    If the graph setting turns out to be NP-hard, then the situation resembles that of the (non-embedded) min-cut linear arrangement problem, which has a polynomial time algorithm for unweighted trees~\cite{Yan85}, but is NP-hard for graphs~\cite{Gav77,GaJo79}.
    
    There are other interesting problems around the complexity and properties of optimal height drawings that might help in finding faster algorithms.
    As one such property, we conjecture that for unweighted trees there always exists an optimal drawing without spiraling edges.
    A spiral on an edge is depicted in Figure~\ref{fig:spiral}.
    There is a more natural notion of quality. In each depth, instead of $\Delta_i$ we take the sequence of heights of lines, ordered in a non-increasing way. We believe drawings which minimize this complexity have the same properties as those that maximize the quality and have minimum height. Some of the proofs even are simpler.

\bibliography{bib}

\end{document}